\documentclass[sigconf]{acmart}

\usepackage{xspace,balance,tabularx,multirow}
\usepackage{flushend}
\usepackage{tikz}
\usetikzlibrary{positioning} % For relative positioning of nodes
\usepackage{pgfplots}
\usetikzlibrary{patterns}
\pgfplotsset{compat=1.16}
\usepackage[ruled, vlined, linesnumbered]{algorithm2e}
\usepackage[noend]{algpseudocode}
\usepackage{xcolor}
\usepackage{float}
\usepackage{colortbl}
\usepackage{hyperref}
\usepackage{bbold}
\usepackage{textcomp}
\usepackage{pifont}
\usepackage{enumitem}
\usepackage{xcolor,framed}
\usepackage[aboveskip=-4pt]{subcaption}
\usepackage[normalem]{ulem}
\usepackage{tikz-3dplot}
\usepackage{bm}
\usepackage[most]{tcolorbox} % for the text boxes
\usepackage{cuted}
\usepackage{caption}
\usepackage{stfloats}

\usepackage{amsthm}

\usepackage{tabularx}

\usepackage[symbol]{footmisc}

\newtheorem{theorem}{Theorem}

\newcommand{\yimin}[1]{{#1}}
\newcommand{\kdd}[1]{{#1}}

\newcommand{\eat}[1]{} %comments out the arument. 
\newcommand{\ie}{{i.e.},\xspace}
\newcommand{\eg}{{e.g.},\xspace}

\newcommand{\stitle}[1]{\noindent{\bf #1.\/}}

\newcommand{\nei}[1]{N(#1)\xspace}
\newcommand{\G}{\mathsf{G}\xspace} % bipartite purchases
\newcommand{\U}{U\xspace} % all users
\newcommand{\Ul}{U_s\xspace} % labeled users
\newcommand{\Uu}{U_r\xspace} % unlabeled users
\newcommand{\V}{V\xspace} % items
\newcommand{\PP}{R\xspace} % persona set
\newcommand{\prs}{r\xspace} % persona
\newcommand{\E}{E\xspace} % edge set of \G
\newcommand{\Pa}{\mathbf{\Phi}\xspace} % persona assignment
\newcommand{\Deg}{d\xspace} % degree
\newcommand{\B}{\tau\xspace} % budget
\newcommand{\score}{s\xspace} % overall sampling score
\newcommand{\Q}{\hat{Q}\xspace} % gt persona distribution
\newcommand{\Qu}{Q\xspace} % user persona distribution
\newcommand{\Mm}{\mathbf{M}\xspace} % matrix example
\newcommand{\Qm}{\mathbf{P}\xspace} % persona distribution in \Pa
\newcommand{\Qmm}{\mathbf{P'}\xspace} % affinity distribution
\newcommand{\stp}{\ell\xspace} % random walk step
\newcommand{\Stp}{\hat{\ell}\xspace} % random walk step upper bound

\newcommand{\Lbl}{\mathbf{L}\xspace} % label matrix constructed from \Pa
\newcommand{\tp}{\pi\xspace} %affinity
\newcommand{\Aff}{\mathbf{\Psi}} % attention matrix
\newcommand{\Att}{\mathbf{\Pi}} % attention matrix (lagency)
\newcommand{\err}{\epsilon\xspace} % err of rrw
\newcommand{\s}{\hat{\Aff}\xspace} % scores
\newcommand{\p}{\mathbf{q}\xspace} % priorities
\newcommand{\w}{w\xspace} % node in \Grrw
\newcommand{\ww}{\hat{w}\xspace}
\newcommand{\Err}{E^*\xspace} % max err in a partite
\newcommand{\NN}{N\xspace} % N=|U|+|V|
\newcommand{\Gtri}{\mathsf{G}'\xspace} % tripartite graph
\newcommand{\h}{h\xspace} %graph conv func
\newcommand{\Etri}{E'\xspace} % edges of tri-graph

\newcommand{\dms}{d\xspace} % embedding dimension
\newcommand{\afff}{\bar{\psi}\xspace} % user-product affinities
\newcommand{\A}{\mathbf{A}\xspace} % adjacency matrix
\newcommand{\X}{\mathbf{X}\xspace} % trainable node embedding matrix
\newcommand{\Z}{\mathbf{Z}\xspace} % node representation matrix
\newcommand{\Atri}{\mathbf{A'}\xspace}
\newcommand{\Xtri}{\mathbf{X'}\xspace}
\newcommand{\Ztri}{\mathbf{Z'}\xspace}
\newcommand{\Ntri}{N'\xspace}

\newcommand{\argmax}[1]{{\operatorname{arg}\,\operatorname{max}}_{#1}\;}

\newcommand{\ours}{\texttt{GPLR}\xspace}
\newcommand{\sample}{\texttt{DUSample}\xspace}
\newcommand{\rw}{\texttt{AffinityCompute}\xspace}
\newcommand{\LLM}{\texttt{LLMAnswer}\xspace}
\newcommand{\rrw}{\texttt{RevAff}\xspace}
\newcommand{\PQ}{\texttt{PQ}\xspace} % priority queue
\newcommand{\lgcntri}{LGCN3\xspace}
\newcommand{\afdtri}{A-LGCN3\xspace}

\newcommand{\mba}{OnlineRetail\xspace}
\newcommand{\ins}{Instacart\xspace}
\newcommand{\insfull}{Instacart Full\xspace}

\newenvironment{customlegend}[1][]{%
    \begingroup
    \csname pgfplots@init@cleared@structures\endcsname
    \pgfplotsset{#1}%
}{%
    \csname pgfplots@createlegend\endcsname
    \endgroup
}%

\def\addlegendimage{\csname pgfplots@addlegendimage\endcsname}
\newcommand{\trref}[1]{Appendix~{#1}}

\AtBeginDocument{%
  }

\copyrightyear{2025}
\acmYear{2025}
\setcopyright{cc}
\setcctype{by}
\acmConference[SIGIR '25]{Proceedings of the 48th International ACM SIGIR Conference on Research and Development in Information Retrieval}{July 13--18, 2025}{Padua, Italy}
\acmBooktitle{Proceedings of the 48th International ACM SIGIR Conference on Research and Development in Information Retrieval (SIGIR '25), July 13--18, 2025, Padua, Italy}\acmDOI{10.1145/3726302.3730118}
\acmISBN{979-8-4007-1592-1/2025/07}
\begin{document}

% \title{Enhancing E-commerce Recommendations: A Persona-based Approach}
\title{You Are What You Bought: Generating Customer Personas for E-commerce Applications}

\author{Yimin Shi}
\email{yiminshi@u.nus.edu}
\orcid{0000-0003-3375-4602}
\affiliation{%
  \institution{National University of Singapore}
  \city{Singapore}
  \country{Singapore}
}

\author{Yang Fei}
\email{yfei11@u.nus.edu}
\orcid{0009-0006-2001-2561}
\affiliation{%
  \institution{National University of Singapore}
  \city{Singapore}
  \country{Singapore}
}

\author{Shiqi Zhang}
\email{shiqi@pyrowis.ai}
\orcid{0000-0002-7155-9579}
\authornote{Shiqi Zhang is the corresponding author.}
\affiliation{%
  \institution{National University of Singapore}
  \institution{PyroWis AI}
  \city{Singapore}
  \country{Singapore}
}

\author{Haixun Wang}
\email{haixun@gmail.com}
\orcid{0000-0002-1378-4241}
\affiliation{%
  \institution{EvenUp}
  \city{San Francisco}
  \country{United States}}

\author{Xiaokui Xiao}
\email{xkxiao@nus.edu.sg}
\orcid{0000-0003-0914-4580}
\affiliation{%
  \institution{National University of Singapore}
  \city{Singapore}
  \country{Singapore}
}

\renewcommand{\shortauthors}{Yimin Shi, Yang Fei, Shiqi Zhang, Haixun Wang and Xiaokui Xiao}

\eat{
\begin{abstract}
    Persona, due to its unique business significance, has received increasing attention in the e-commerce area. The emergence of Large Language Models (LLM) is bringing new opportunities for the design of persona-based recommender systems. To our surprise, LLM-assigned user persona labels can easily improve major accuracy metrics such as F1-score@K and NDCG@K by 5-15\% when intuitively combined with the state-of-the-art graph convolution-based recommendation methods. However, due to the high costs of LLMs, labeling personas for all users can lead to significant economic expenses in practice. In our work, we further propose an effective algorithm to predict unlabeled users' persona based on a small subset of label ground truth. The algorithm contains two main steps: (1) diversity and uncertainty-aware user sampling, and (2) a bipartite random walk-based persona prediction algorithm. Our experiments show that applying our predicted persona for unlabeled users can also significantly improve the recommendation accuracy metrics, sometimes even outperforming the fully labeled version.    
\end{abstract}}

\begin{abstract}

    In e-commerce, user representations are essential for various applications. Existing methods often use deep learning techniques to convert customer behaviors into implicit embeddings. However, these embeddings are difficult to \yimin{understand} and integrate with external knowledge, limiting the effectiveness of applications such as customer segmentation, search navigation, and product recommendations. To address this, our paper introduces the concept of the \textit{customer persona}. Condensed from a customer's numerous purchasing histories, a customer persona provides a multi-faceted and \yimin{human-readable} characterization of specific purchase behaviors and preferences, such as \textit{Busy Parents} or \textit{Bargain Hunters}.
    
    This work then focuses on representing each customer by multiple personas from a predefined set, achieving readable and informative explicit user representations. To this end, we propose an effective and efficient solution \ours. To ensure effectiveness, \ours leverages pre-trained LLMs and few-shot learning to infer personas for customers. To reduce overhead, \ours applies LLM-based labeling to only a fraction of users and utilizes a random walk technique to predict personas for the remaining customers. To further enhance efficiency, we propose an approximate solution called \rrw for this random walk-based computation. \rrw provides an absolute error $\err$ guarantee while improving the time complexity of the exact solution by a factor of at least $O\left(\frac{\err\cdot|\E|\NN}{|E|+\NN\log\NN}\right)$, where $N$ represents the number of customers and products, and $E$ represents the interactions between them. \yimin{We evaluate the performance of our persona-based representation in terms of accuracy and robustness for recommendation and customer segmentation tasks using three real-world e-commerce datasets.} Most notably, we find that integrating customer persona representations improves the state-of-the-art graph convolution-based recommendation model by up to 12\% in terms of NDCG@K and F1-Score@K.

\end{abstract}


% \keywords{Recommendation, Large Language Model, Data Mining}

% date things

\maketitle

\section{Introduction}

Large Language Models (LLMs), such as GPT-4~\cite{achiam2023gpt}, hold substantial potential for various e-commerce applications. For example, in digital bookstores, LLMs can recommend books tailored to a customer’s specific interests in various genres and authors, while also explaining the rationale behind each recommendation. Furthermore, LLMs can comprehend literary nuances and character relationships within each book, aiding in the construction of a knowledge base. When combined with retrieval-augmented generation techniques, LLMs can leverage this knowledge base to deliver more personalized and detailed product QA for the books.

The aforementioned applications of LLMs, however, require {\it explicit} customer representations that characterize the customers in natural language (\eg ``Bargain Hunter,'' ``Audiobook Listeners''), so as to provide LLMs with useful information regarding potential customer preferences. In contrast, existing e-commerce recommendation methods (\eg ~\citeN{wang2021ebay, he2020lightgcn, yu2022low}) predominantly rely on {\it implicit} customer representations. In particular, these methods map each customer's personal data (\eg browsing histories, clicks, and purchases) to a fixed-size numeric vector, which cannot be easily \yimin{understood} by LLMs or even human experts.
One possible solution is to conduct user surveys to gather explicit user preferences, but such surveys are time-consuming and unrepeatable, and often result in incomplete data. Some recent studies~\citeN{yu2023folkscope,yu2024cosmo, li2015mining} propose representing customers based on their preferred product categories. Nevertheless, these representations are insufficiently informative as they fail to capture detailed purchase behaviors.

To address the above deficiencies, we introduce the concept of \textit{Personas} as a new dimension to describe customers' purchase behaviors and product preferences.
For instance, on an e-commerce platform, customer personas may include \textit{Bargain Hunter}, \textit{Brand Loyalist}, \textit{Health Enthusiast}, and \textit{Tech Savvy}, among others. Each persona is accompanied by a detailed definition that provides insights into a customer's purchasing characterization. This allows each persona to be multi-faceted and not confined to product category granularity. For example, a \textit{Health Enthusiast}, whose core pursuit is achieving and maintaining good health, may have a subtle preference toward different products in the \textit{fruits} category: kiwis are favored for their rich vitamin content, while high-calorie options such as bananas are avoided.
Unlike event-driven purchase intentions that capture a customer's immediate motivations~\citeN{yu2023folkscope,yu2024cosmo}, personas are more durable and reusable, as they often relate to personal interests, habits, and lifestyles that do not frequently change.
Therefore, assigning each customer multiple personas based on her historical purchase behaviors provides a powerful explicit representation that enhances the capabilities of LLMs in e-commerce.

To generate personas for each customer, we propose utilizing LLMs themselves. By instructing LLMs to summarize the customers' purchase behaviors using natural language, we can effectively map each customer to her persona labels, leveraging the LLMs' knowledge and reasoning capabilities.
However, scalability presents a substantial challenge if we are to repeatedly \kdd{label} numerous users on a real-world e-commerce platform, in response to the continuous update of user engagement data. For an e-commerce platform with 10 million active customers, relabeling all users monthly based on their dynamic purchase statistics with GPT-4 costs approximately 2.4 million dollars per year.

% 5. our approach
To address this scalability challenge, we propose an efficient and cost-effective solution, dubbed \ours, for generating customer persona representations on real-world e-commerce platforms.
\ours first leverages LLMs to label a small, carefully selected user set within a limited budget.
LLMs are instructed to assign personas for each customer from a platform-specific prototype persona set.
Subsequently, \ours utilizes random walks to analyze purchase behavior similarities among users, based on which it infers the personas for unlabeled users through a weighted aggregation of prototype persona labels. We further propose \rrw to reduce the time complexity of random walk-based computations by approximating exact random walk probabilities with a theoretical error tolerance.

Given that personalized product recommendation is a key downstream application of the explicit customer representation, we describe a methodology for integrating customer personas as an additional partition into existing graph convolution-based methods in Section \ref{sec:ProductRecommendation}.
Based on this methodology, we integrate personas with two state-of-the-art recommendation models~\cite{yu2022low, wu2024afdgcf}, and extensively compare them with six cutting-edge recommendation models. Specifically, we find that persona-enhanced models can outperform competitors on three real-world datasets with up to 13 million user-item interactions. Notably, personas can improve upon the original model by up to 12\% in NDCG@K and F1-Score@K.
\yimin{Besides recommendation, we further demonstrate that personas outperform existing explicit customer representations in customer segmentation in terms of robustness and cluster quality.}
In addition, extensive experiments show that \rrw can process large-scale datasets in tens of seconds, achieving empirical errors significantly lower than the theoretical bound.

% 6. contributions
To summarize, we make the following contributions:
\begin{itemize}[topsep=2pt,itemsep=1pt,parsep=0pt,partopsep=0pt,leftmargin=11pt]
    \item We propose representing customers through personas, enabling the integration of external agents' knowledge and reasoning.
    \item We propose \ours, which significantly reduces the cost of using LLMs when generating effective customer personas.
    \item We enhance the scalability of \ours by proposing \rrw to approximate its random walk-related computations efficiently.
    \item We present a methodology for seamlessly combining personas with graph convolution-based recommendation models.
    \item We conduct comprehensive experiments to show the superiority of the proposed solutions in \yimin{product recommendation and customer segmentation}.
\end{itemize}
\vspace{-2mm}

% Section: Preliminary
\section{Problem Formulation}
This section first introduces the bipartite graph representation of customer purchase histories, and then proposes the concept of personas and defines our representation generation problem, after which we discuss three important downstream applications.

\vspace{-1mm}
\subsection{Preliminaries}\label{sec:notations}
\stitle{Purchase histories}
{We represent purchase histories of users on an e-commerce platform as a bipartite graph $\G=(\U, \V, \E)$, where $\U$ and $\V$ are sets of nodes representing users and products on the e-commerce platform, respectively.} The edge set $E$ contains the purchase history $e_{ij}=(u_i, v_j)\in \E$, representing that the user $u_i\in U$ purchased the product $v_j\in V$ in the past. For an edge $e_{ij}=(u_i, v_j)\in \E$, we say $u_i$ and $v_j$ are neighbors. We denote $\nei{u_i}$ (resp.\ $\nei{v_j}$) as the neighbor set of $u_i$ (resp.\ $v_j$) and represent its degree as $\Deg_i$ (resp.\ $\Deg_j$).

\stitle{Notations}
Throughout this paper, we denote matrices in bold uppercase, \eg $\Mm$. We use $\Mm[v_i]$ to denote the $v_i$-th row vector of $\Mm$, and $\Mm[:,v_j]$ to denote the $v_j$-th column vector of $\Mm$. In addition, we use $\Mm[v_i,v_j]$ to denote the element at the $v_i$-th row and $v_j$-th column of $\Mm$.
Given an index set $S$, we let $\Mm[S]$ (resp.\ $\Mm[:,S]$) be the matrix block of $\Mm$ that contains the row (resp.\ column) vectors of the indices in $S$. Table~\ref{tab:notations} lists the frequently used notations.

\vspace{-1mm}
\subsection{Customer Personas}\label{sec:persona}
The concept persona refers to the social face that an individual presents to the world from a psychological perspective~\cite{jung2014two}.
In this work, we introduce the concept of \textit{customer persona} for customers on e-commerce platforms. Specifically, each customer persona summarizes a characterization of a customer's specific purchase behaviors and product preferences, such as \textit{Bargain Hunters}, \textit{Health Enthusiasts}, \textit{Tech Savvy}, and \textit{Busy Parents}. Additionally, each customer persona is associated with a detailed explanation. For example, the \textit{Busy Parents} persona describes an individual who frequently purchases kid-friendly products, diapers, baby food, and other family necessities, often seeking convenience through pre-made meals and time-saving products. The {\textit{Bargain Hunters} represents customers who are always looking for the best deals and discounts, buying sale products with coupons and purchasing in bulk to save money.} Advantages of the customer persona are three-fold.

\stitle{Informativeness}
In contrast to conventional features such as demographic information, the customer persona demonstrates high informativeness by seamlessly integrating multi-dimensional information from the user's purchase preference. For example, the Busy Parent persona captures the information from the perspectives of baby needs, household, and convenience.

\yimin{\stitle{Readability} Although deep learning techniques can condense information into embeddings, such representations often lack human-readable dimensions that convey clear and understandable messages. In contrast, customer personas, along with their corresponding detailed explanations, are easy for both humans and LLMs to read and comprehend.
}

\stitle{Robustness}
The customer persona is tied to a user's values, habits, and living conditions, which do not change frequently. This makes it a suitable choice as a relatively long-standing label for customers.

Notably, in this work, we assume that the e-commerce platform has predefined a set of customer personas, denoted as $\PP$. 
To establish a set of comprehensive and representative customer personas, the platform can leverage LLMs associated with prompt engineering techniques\footnote{All related prompts in this work will be provided in \trref{\ref{app:b}}.}.
Specifically, the service provider supplies the LLM with a portion of real-world customer purchase histories to generate an initial set of persona candidates, which are further refined and deduplicated by the LLM. Additionally, the predefined personas can be validated based on customer feedback on the platform, finally yielding a proprietary persona set for the platform.

\begin{table}[!t]
\centering
\renewcommand{\arraystretch}{1.1}
\begin{small}
\caption{Frequently used notations.}\vspace{-3mm} \label{tab:notations}
\resizebox{\columnwidth}{!}{
\begin{tabular}{rp{2.7in}}	
    \toprule
    \bf Notation & \bf Description \\
    \midrule
    $\G = (\U, \V, \E)$ & A bipartite graph with user node set $\U$, item node set $\V$ and historical purchase edge set $E$.\\
    $\PP$ & Predefined persona set.\\
    $\Uu, \Ul$ & Unlabeled user set and prototype user set.\\
    $\B$ & LLM query budgets and $\B=|\Uu|$.\\
    $\Pa$ & Customer representation matrix by personas.\\
    $\Aff$ & The user-persona affinity matrix in Eq.~\eqref{eq:aff}.\\
    $\Stp$ & Random walk step in Eq.~\eqref{eq:exactaff}. \\
    $\err$ & The absolute error in Definition~\ref{def:e-approximate}.\\
    
    \bottomrule
\end{tabular}
}
\end{small}
\vspace{-2mm}
\end{table}

\subsection{{Customer Representation by Personas}}
{To personalize each customer with better informativeness, \yimin{human-readability,} and robustness, we target to assign each customer multiple personas from the predefined persona set that align with her historical purchase behaviors. We formalize the major problem to solve in this paper as follows.}

\stitle{Goal}
Given a bipartite graph $\G = (\U, \V, \E)$ representing user purchase histories and a predefined customer persona set $\PP$, this work focuses on finding a customer representation matrix $\Pa\in \{0,1\}^{|\U|\times |\PP|}$. $\Pa[u_i]$ represents the persona representation of user $u_i$, where $\Pa[u_i, \prs_k]=1$ indicates that the customer $u_i$ possesses the persona $\prs_k$; otherwise, $\Pa[u_i, \prs_k]=0$. 
In other words, the core of this work lies in effectively and efficiently assigning personas to each customer.

\subsection{Downstream Applications}\label{sec:app}
The persona-based representation can be applied across a range of downstream tasks. For service providers and retailers, personas offer a new and informative perspective for customer segmentation. For the customer experience, the personas not only serve as filter labels in searches but also facilitate product recommendations. In what follows, we will illustrate these application scenarios in detail.
% wind2008market,

\stitle{Customer segmentation~\cite{wan2022fast,pai2022unsupervised}}
Customer segmentation is the process of dividing customers into distinct groups based on similarities in specific attributes such as demographic features and interests.
The proposed personas offer a valuable method for categorizing customers according to their spending habits and purchasing behavior patterns, thereby enhancing data analytics in advertising, marketing, and customer relationship management. Furthermore, the customer persona set $\PP$ provides $2^{|\PP|}$ distinct combinations of persona for intersecting or uniting customer groups, allowing analytics to vary from the coarsest to the finest granularity. In addition to enhancing analytics, this persona-based approach also facilitates the design of personalized strategies to increase customer engagement and business profitability.

\stitle{Customer-centric search navigation~\cite{yu2023folkscope, yu2024cosmo}}
{On the e-commerce platform, customers search for products by entering concepts or keywords. In response, the platform returns a list of relevant products along with a navigation bar that enables customers to explore these products based on a taxonomy focused on the catalog and metadata of the products. This taxonomy can be enhanced by incorporating the proposed personas as an additional layer at the highest level, creating a new taxonomy that integrates both customer-centric persona concepts and the product-focused hierarchy. For example, the persona Busy Parent could serve as a supernode in the new taxonomy, encompassing three departments: Baby, Health \& Household, and Home \& Kitchen, based on the Amazon catalog. This structure allows the taxonomy to reflect both customer preferences and product categories, facilitating a more intuitive and targeted shopping experience.}

\stitle{Product recommendation~\cite{he2020lightgcn,wu2024afdgcf}}
Customer representation by personas can enhance the effectiveness of personalized product recommendations.
Specifically, these representations can serve as $|\PP|$ additional input features for well-adopted recommendation models, such as those based on graph neural networks. In contrast to merely considering historical transactions, incorporating personas provides more aggregated insights w.r.t.\ each customer.
Moreover, as comprehensible features, they can be integrated with external knowledge, such as the favorite products of each persona, thereby improving the performance of recommender systems.
In addition, personas can also alleviate the cold-starting problem. To explain, new customers can select personas that fit their own profiles when they first join, and retailers of new products can choose personas that align with the target customers for these products. Hence, new items can be recommended to users with matching personas, and vice versa.

% Section: Related work
\section{Related Work}
In this section, we briefly review existing works focused on generating implicit and explicit user representations.

Most previous methods represent users as low-dimensional embeddings optimized according to specific objective functions (see Ref.~\cite{khoshraftar2024survey} for the latest survey).
Among them, numerous solutions~\cite{koren2009mf, he2017ncf, he2020lightgcn, wang2021ebay, zhang2024scaling,song2023xgcn,chen2023clustered,zhang2024linear,wu2024afdgcf} are based on collaborative filtering, ranging from traditional matrix factorization solution~\cite{koren2009mf} to the modern approaches based on graph neural networks~\cite{he2020lightgcn,zhang2024scaling,song2023xgcn,berg2017graph,choudhary2024interpretable,yu2020graph,yu2022low,zhang2024linear,wu2024afdgcf}. However, these embeddings are implicit representations, meaningful only within the recommender system, and difficult for humans or LLMs to understand. Similar situations happen in other personalized e-commerce applications, such as customer segmentation~\cite{wan2022fast,pai2022unsupervised} and product search~\cite{ai2017learning, ai2019zero, yu2023folkscope, yu2024cosmo,choudhary2024interpretable}. In addition, although some works~\cite{barkan2020explainable, banerjee2016capres, deng2022toward} utilize the concept of personas in their modeling process, they still represent each persona as a set of implicit features instead of a human-readable definition.

Some works attempt explicit customer representation by utilizing the concept of persona, but fundamentally differ from our approach. For example, \citet{li2015mining} posteriorly assigns each user with a single persona label based on product categories rather than customer-centric characterization.
\citet{elad2019learning} models customer personalities that are not directly related to e-commerce.
Additionally, another line of related works, such as FolkScope~\cite{yu2023folkscope} and COSMO~\cite{yu2024cosmo}, leverage LLMs to construct a commonsense knowledge graph for applications in the e-commerce field. In this graph, the head entity represents a pair of co-buy or view-buy items from the same category, while the tail entity explains the possible reason for this co-occurrence. In other words, this knowledge graph can currently serve only as a textual description for item pairs from the same category. However, it remains an open question how to convert this into an explicit and extensive representation for each item and, further, for each customer.

% Section: Methodology
\section{Solution Framework}
{In this section, we introduce the framework of \ours, which Generates customers' Persona representation matrix $\Pa$ through leveraging Large language models and Random walk-based affinities.}

\subsection{Main Idea}\label{sec:main-idea}
Benefiting from the finding that LLMs like GPT-4 are reliable in understanding and answering questions related to the shopping domain~\cite{sun2024large}, a straightforward and effective approach to generate $\Pa[u_i]$ for a user $u_i$ is to serialize all purchase histories of $u_i$ and ask a pre-trained LLM which personas $u_i$ belongs to.
However, due to their large-scale model structure and autoregressive generation mechanism, LLMs incur a prohibitively immense computational overhead when generating personas for millions of customers on real-world e-commerce platforms, regardless of whether using online APIs or deploying locally. Furthermore, the need to repeatedly relabel each customer due to the dynamic nature of their personas further amplifies this problem.

To be more cost-effective, we propose the solution \ours, whose idea is to sample a small fraction of users $\Ul$ as \textit{prototype users} and generate their customer representations using LLMs, then infer personas for remaining unlabelled users, $\U \setminus \Ul$, based on their proximity to the prototype users within $\G$.
The pseudocode of \ours is illustrated in Algorithm~\ref{alg:overview}. In particular, \ours takes as inputs the purchase histories $\G$, the persona set $\PP$, a budget constant $\B$ indicating the cardinality of the sampled user set $|\Ul|$, the number of iterations $T$, and a cutoff constant $k$. {Algorithm~\ref{alg:overview} initializes $\Pa$ and $\Aff$ to zero matrices and $\Ul$ to an empty set. $\Aff \in \mathbb{R}^{|U| \times |\PP|}$ is called the \textit{user-persona affinity} matrix, which is the core of the subsequent subroutines. (For ease of presentation, we defer the formal definition and computation of $\Aff$ to Section~\ref{sec:PersonaApproximation}.) Intuitively, $\Aff[u_i, r_j]$ is a weighted aggregation of the label associated with $r_j$ w.r.t.\ prototype users in $u_i$'s vicinity.} In Lines 2-6, Algorithm~\ref{alg:overview} repeats over $T$ times. In each iteration $t$, it invokes a sampling strategy called \sample to return a set of users $\U_t$ with $|\U_t|=\B/T$, followed by a subroutine called \LLM that labels $\U_t$ to personas in $\PP$ based on LLMs. 
{After that, \ours updates the prototype user set $\Ul$ and recomputes $\Aff$ based thereon by invoking \rw. In Line 7, for each remaining user $u_i\in \U \setminus \Ul$, we select the $k$ personas with the largest affinity scores from the latest $\Aff[u_i]$ as the persona representation of $u_i$.}

\subsection{Subroutine Descriptions}\label{sec:overview-subsoutine}
In what follows, we elaborate on the three subroutines \sample, \LLM, and \rw in \ours. 

\begin{algorithm}[t]
\caption{\ours}\label{alg:overview}
\KwIn{Purchase histories $\G$, persona set $\PP$, budget constant $\B$, number of iterations $T$ and a cutoff constant $k$}
\KwOut{Customer representation by personas $\Pa$}
$\Pa,\Aff \gets \{0\}^{|\U|\times |\PP|}$\ ; $\Ul \gets \emptyset$;\\
\For{$t\gets 1$ \textbf{to} $T$}{
$\U_t \gets \sample(\G, \Pa, \Aff, \B/T, t)$;\\
$\Pa[\U_t] \gets \LLM(\U_t, \PP)$;\\
$\Ul \gets \Ul\cup\U_t$;\\
$\Aff \gets \rw(\G, \Pa, \Ul)$;\\
}
For each $u_i\in \U \setminus \Ul$, set each $\Pa[u_i,r_w]$ to 1 where $r_w$ is among the top-$k$ personas of $\Aff[u_i]$;\\
\Return{$\Pa$;}
\end{algorithm}

\stitle{\sample}
We first introduce the Diversity-Uncertainty (\texttt{DU}) sampling approach, dubbed as $\sample$, {which considers both persona diversity and user uncertainty based on the current persona affinity scores,} with the following motivations. From the perspective of diversity, we observe that the distribution of persona labels is often severely biased in real-world datasets, \ie the majority of users are associated with a small set of dominating personas. For example, in the \ins dataset where customer representations are labeled from 51 expert-designed personas using LLM, the two most popular personas are assigned to $76.2\%$ and $75.8\%$ of users, respectively, while the 20 least common personas together cover only 7.4\% of users. In other words, randomly selecting users for the LLM labeling inherits the bias in the collected labels, leading to false positive associations for the remaining users. {To mitigate this bias, we consider labeling users more likely to contain less common personas,} ensuring better diversity in the collected persona labels. Regarding uncertainty, similar to the labeling process in active learning~\cite{ren2021active}, we want to label the users who are more uncertain in the following \rw routine, \ie having higher probabilities of being assigned with wrong persona labels. Instead of postponing the decision-making, directly labeling these users with LLMs increases the accuracy in expectation.

Based on the aforementioned insights, $\sample$ randomly selects $\B/T$ users from $\U$ in the initial step ($t=1$) to form $\U_1$. At each subsequent time step ($t > 1$), we compute the \texttt{DU} score for all $u_i\in \Uu$ as follows, and selects the top-$(\B/T)$ users with the highest \texttt{DU} scores among all unlabeled users.
\begin{equation}\label{eq:duscore} 
\textstyle\score(u_i) = \sum\limits_{\prs_m \in \PP} \Q(\prs_m)\cdot\log \left( \frac{\Q(\prs_m)}{\Qu_i(\prs_m)}\right) -\Qu_i(\prs_m)\cdot\log \Qu_i(\prs_m)
\end{equation}
{In Eq.~\eqref{eq:duscore}, $\Q$ is defined as the persona label distribution in the currently collected $\Pa$ as follows: 
$$ \textstyle \Q(\prs_m) = \frac{\sum\limits_{u_i\in \Ul} \Pa[u_i,\prs_m]}{\sum\limits_{u_i\in \Ul,\prs_n\in \PP}\Pa[u_i, \prs_n]},$$
and $\Qu_i$ represents the normalized user-persona affinity distribution of user $u_i$ in the current iteration, defined as:
$$ \textstyle \Qu_i(\prs_m) = \frac{\Aff[u_i, \prs_m]}{\sum\limits_{\prs_n\in \PP}\Aff[u_i,\prs_n]}.$$
The first term of \texttt{DU} score is actually the KL divergence between $\Q$ and $\Qu_i$. Selecting users with higher divergence values is expected to enhance the persona diversity in the collected prototype labels. The intuition is that if $u_i$'s user-persona affinity distribution significantly differs from the currently collected persona distribution, it indicates that $u_i$'s purchase behavior is more likely to differ from the majority of users. The second term represents the entropy of $\Qu_i$. Labeling users with high entropy is expected to improve the accuracy of the persona assignment. The intuition here is that a lower entropy suggests that $u_i$'s affinities on different personas are similar, indicating that the persona information about this user is still ambiguous and has high uncertainty in determining her labels.}

\stitle{\LLM}
Subsequently, we send these users' purchase logs and predefined persona set $\PP$ to the LLM and update its association results in $\Pa$. 
We denote this approach as $\LLM$, which takes as inputs the bipartite graph $\G$ and the persona set $\PP$.
Specifically, we first employ serialization to transform each customer's purchase logs into a natural language format. We then input the transformed result, a predefined persona set $\PP$, instruction prompts, and in-prompt examples into the LLM. These in-prompt examples are defined to improve the generation quality and specify the output format. With few-shot learning, the pre-trained LLM can leverage its extensive real-world knowledge and reasoning abilities to associate the user with the relevant personas by referencing her purchased product names and amounts. 

\stitle{\rw}
{Due to the update of prototype user set $\Ul$ and their persona representations $\Pa[\Ul]$, we conduct the random walk-based solution $\rw$ to recompute the user-persona affinity matrix $\Aff$. 
Based on the homophily principle~\cite{easley2010networks}, we assume that a user is likely to share identical personas with other users who exhibit similar purchasing behaviors, referred to as \textit{neighbor} users. To compute $\Aff$, for each user $u_i$, our main idea is to identify the labeled neighbor users in the vicinity and generate $\Aff[u_i]$ according to the aggregation of their persona representations.}
\section{User-Persona Affinity Computation}\label{sec:PersonaApproximation}
This section first describes the original computation process of the user-persona affinity matrix $\Aff$. It then introduces \rrw as an efficient approximate solution based on reverse updating.

\subsection{Exact Solution}
{
To capture structural closeness between users, \kdd{we define the \textit{attention} as the mean of probabilities of a random walk, starting from $u_i$ reaching $u_k$ within $\Stp$-steps or fewer, \ie $\frac{1}{\Stp}\sum_{\stp\leq\Stp}\tp_{\stp}(u_i, u_k)$.} In particular, for a given $u_i$, an $\stp$-step random walk starts from $u_i$, and at each step ($\leq\stp$), it first navigates to a random product purchased by the current user node and then randomly jumps to a user node that purchased this product. \kdd{The attention matrix $\Att\in \mathbb{R}^{|U|\times|U|}$ is
\begin{equation}\label{eq:exactaff}
    \textstyle \Att= \frac{1}{\Stp}\sum\limits_{\stp\leq\Stp}\Att_{\stp}=\frac{1}{\Stp}\sum\limits_{\stp\leq\Stp}\mathbf{I}\cdot (\Qm \cdot \Qmm)^{\stp},
\end{equation}
where $\Att_{\stp}[u_i,u_j]=\tp_{\stp}(u_i, u_j)$ for every $u_i,u_j\in \U$.}
In Eq.~\eqref{eq:exactaff}, $\Qm \in \mathbb{R}^{|U|\times|V|}\ $ and $\Qmm \in \mathbb{R}^{|V|\times|U|}$ are two transition matrices of $\G$, where $\Qm[u_i,v_j]=\frac{1}{|\nei{u_i}|}$ and $\Qmm[v_j,u_i] = \frac{1}{|\nei{v_j}|}$ if $e_{ij}\in \E$, and $\Qm[u_i,v_j]=\Qmm[v_j,u_i]=0$ otherwise.}

{After that, we define a matrix $\Lbl \in \mathbb{R}^{|\U|\times|\PP|}$ to measure the relative importance of each persona w.r.t.\ each user node. Specifically,
\begin{equation}\label{eq:labelmatrix} 
  \textstyle \Lbl[u_i,\prs_m]=\frac{c_m\cdot \Pa[u_i, \prs_m]}{\sum\limits_{\prs_n\in \PP}\Pa[u_i, \prs_n]},\quad c_m = \left( \frac{\min\limits_{\prs_n\in \PP} \Q(\prs_n)}{\Q(\prs_m)} \right)^{\beta}  
\end{equation}
if $u_i\in \Ul$ is a prototype user and $\Lbl[u_i,\prs_m]=0$ otherwise. In Eq.~\eqref{eq:labelmatrix}, a coefficient $c_m$ and a hyperparameter $\beta$ are considered to deal with the persona distribution bias, such that the signal of minor personas can be emphasized. Finally, the user-persona affinity matrix $\Aff$ is defined as}:
\begin{equation}\label{eq:aff}
    \textstyle\Aff = \mathbf{\Pi} \cdot \Lbl,
\end{equation}
where $\Aff[u_i, \prs_m]$ represents the user-persona affinity from user $u_i\in \U$ to persona $r_m\in \PP$.

{Based on the aforementioned definition, the exact solution for computing the matrix $\Aff$ involves first calculating the matrices $\Att$ and $\Lbl$, and then determining $\Aff$ according to Eq.~\eqref{eq:aff}. Within $\Aff$, a submatrix $\Aff[\U\setminus\Ul]$ is utilized for Algorithm \ref{alg:overview}. The time complexity of this exact computation is as follows \footnote{\yimin{Detailed proofs will be provided in \trref{\ref{app:a}}.}}.

\begin{theorem}\label{thm:1} 
The time complexity for computing the exact $\Aff$ by Eq.~\eqref{eq:aff} is $O(|E|\cdot|\U| + (\Stp-1)|\U|^3 + |\PP|\cdot|\U|^2)$.
% With $|E|$ represents the total number of edges in $G$.
\end{theorem}

{According to Theorem~\ref{thm:1}, this exact solution fails to handle large $\G$ for two main reasons. \kdd{First, due to the sparsity of $\Lbl$, only attention values to prototype users (\ie $\Att[:, \Ul]$) are utilized in computing $\Aff$, however, $\Ul$ only represents a small fraction of $\U$, \eg $\tau=10\%$, incurring a large number of excessive computations.} Second, due to the high connectivity in $\G$, $\Att_\stp$ is diminishing rapidly as the step $\stp$ increases, resulting in a multitude of tiny attention values that can be disregarded.}

\subsection{Fast Approximation}\label{sec:ReverseRandomWalk}
{To address the efficiency issues mentioned above, we propose an approximation method called \rrw. At a high level, for a given persona $r_m$, \rrw performs $\ell$-step random walk in a reverse manner from every prototype user to compute $\Att_\stp[:, \Ul]$, which avoids the materialization of the entire $\Att_\stp$. Due to the existence of massive tiny attention values, \rrw resorts to computing $\err$-approximate user-persona affinities $\hat{\Aff}[:,r_m]$ for an input persona $r_m\in \PP$, which is defined as follows.
\begin{definition}[$\err$-approximate user-persona affinity]\label{def:e-approximate}
    Given an absolute error threshold $\err$, a persona $r_m$, and a prototype user set $\Ul$, for any $u_i\in \Ul$, $\hat{\Aff}[u_i,\prs_m]$ is an $\err$-approximation of $\Aff[u_i,\prs_m]$ if it satisfies $|\Aff[u_i,\prs_m] - \hat{\Aff}[u_i,\prs_m]| \leq \err$.
\end{definition}
}
For ease of presentation, we only illustrated the single persona case in the sequel.
To compute $\hat{\Aff}$, we can repeatedly invoke \rrw by taking as an input every persona $r_m\in \PP$.

\begin{algorithm}[t]
\caption{\rrw}\label{alg:rrw}
\KwIn{Purchase histories $\G$, target persona $\prs_m$, de-bias coefficient $c_m$, prototype representation matrix $\Pa$, and error tolerance $\err$}
\KwOut{$\hat{\Aff}[:,r_m]$}
$\s_t[:,\prs_m],\p_t \gets \mathbf{0},\ \forall t\in \{2\Stp\}$;\\
% \yimin{$\s_0[\prs_m], \p_0[\prs_m]\gets c_m$;}\\
$\s_0[:,\prs_m], \p_0\gets \Lbl[:,\prs_m]$ in Eq.~\eqref{eq:labelmatrix};\\
% $\PQ = \texttt{PriorityQueue}(\W,\p)$;\\
\While{$\exists \ww,~\forall t\leq 2\Stp$, $\p_t[\ww] \geq \epsilon/2\Stp$, }{
    $t,\ww \gets \argmax{t,\ww\in \U \cup \V}{\p_t[\ww]}$; \\
    \For{$\w \in \nei{\ww}$}{
        $\Delta \gets \frac{\p_t[\ww]}{|\nei{\w}|}$;\\
        $\s_{t+1}[\w,\prs_m] \gets \s_{t+1}[\w,\prs_m] + \Delta$;\\
        $\p_{t+1}[\w] \gets \p_{t+1}[\w] + \Delta$;\\
    }
    $\p_t[\ww] \gets 0$; \\
}
% \Return{$\hat{\Aff}[:,r_m]\gets\left(\s_{2\stp}[U,\prs_m]\right)^{\mathsf{T}}$;}
\Return{$\hat{\Aff}[:,r_m]\gets \frac{1}{\Stp}\sum_{\stp \leq \Stp} \s_{2\stp}[U,\prs_m]$;}
\end{algorithm}

Given a target persona node $\prs_m$, the main idea of \rrw is to reversely conduct a deterministic graph traversal and compute the estimated $\Aff[:,r_m]$.
We use $\s_t[:,\prs_m]\in \mathbb{R}^{|U|+|V|}$ to record the estimation for every node $\w\in \U\cup \V$ to $r_m$ in the hop $t=1,2,...,2\Stp$, \kdd{where a hop means a user node jumps to an item node in a specific step, or vice versa.}
Following the previous work~\cite{lofgren2013reverse}, we have the following recurrence formula for designing \rrw:
\begin{equation*}\label{eq:rec}
    \textstyle\Aff_t[\w,\prs_m]= 
    \begin{cases}
   \Lbl[w,\prs_m], & t=0 \text{ and } w\in \U \\    
  \sum\limits_{\ww \in \nei{\w}}\frac{\Aff_{t-1}[\ww,\prs_m]}{|\nei{\w}|},  & \text{otherwise} 
\end{cases}.
\end{equation*}

Based on Eq.~\eqref{eq:exactaff} and Eq.~\eqref{eq:aff}, we derive that the target $\Aff$ is given by $\frac{1}{\Stp}\sum_{\stp\leq \Stp} \Aff_{2\stp}$. The pseudocode of \rrw is shown in Algorithm~\ref{alg:rrw}.
In Lines 1-2, we first initialize $\s_0$ and its corresponding temporary vector $\p_0$ to zero vectors, and then set  $\s_0[u_k,\prs_m]=\p_0[u_k,\prs_m]=\Lbl[u_k,\prs_m]$ for all $u_k \in \Ul$. 
In subsequent Lines 3-9, we iteratively estimate the desired values of each hop in an asynchronous manner.
Specifically, in each iteration, we select the largest value in all $\p_t$ vectors with $t=1,2,...,2\Stp$, and denote the corresponding node and hop as $\ww$ and $t$, respectively. For every neighbor $w \in \nei{\ww}$, we increase the value of $\s_{t+1}[w,\prs_m]$ and  $\p_{t+1}[\w]$ by a term $\Delta=\frac{\p_t[\ww]}{|\nei{\w}|}$. After that, we reset $\p_{t+1}[\ww]$ to $0$ since all its accumulated updates have been propagated backward. Algorithm~\ref{alg:rrw} repeats this process until the values in all $\p_t$ vectors with $t=1,2,...,2\Stp$ are lower than a predefined threshold $\err/2\Stp$. \kdd{At last, we assign $\frac{1}{\Stp}\sum_{\stp \leq \Stp} \s_{2\stp}[U,\prs_m]$ to $\hat{\Aff}[:,r_m]$ as the output.}
The following theorems show the theoretical guarantee and time complexity that \rrw provides.

\begin{theorem}
\label{thm:2}
    Given a persona $\prs_m$ and an error $\err$, $\s[:,\prs_m]$ generated by \rrw is an $\err$-approximate user-persona affinity.
    % has the following constraint on the absolute error: $|\tp_\stp(u,\prs)-\s_\stp(u)|\leq \err,\forall u\in \Uu$.
\end{theorem}

\begin{theorem}
\label{thm:3}
    Given a graph $\G$, a persona set $\PP$, and an error $\err$, the time complexity of estimating $\Aff$ by invoking \rrw from every $r_m\in \PP$ is $O\left(\frac{1}{\err}\Stp^2\left(\NN\log(\Stp\NN)+|E|\right)\right)$, where $N=|U|+|V|$.
\end{theorem}
By Theorems~\ref{thm:1} and~\ref{thm:3}, \rrw improves the running time of the exact solution by a factor of $O\left(\frac{\err\cdot |\E|\NN}{|\E|+\NN\log\NN}\right)$ for $\Stp=1$ and $O\left(\frac{\err\cdot \NN^3}{\Stp^2(|\E|+\NN\log\NN)}\right)$ for $\Stp>1$, demonstrating the efficiency of \rrw.
\section{Product Recommendation}\label{sec:ProductRecommendation}
In this section, we first take graph convolution-based collaborative filtering algorithms~\citeN{he2020lightgcn, yu2022low} as an example and review their general framework. We then present our approach that leverages the proposed users' persona-based representations to enhance these methods for product recommendation.
\eat{Section \ref{sec:ProductRecommendation}, we discuss integrating customer representations by personas into existing graph convolution-based algorithms to enhance personalized recommendation performance.}

\stitle{Graph convolution-based methods}
Given the bipartite graph $\G$, these methods represent all user nodes and item nodes as trainable embeddings $\X\in \mathbb{R}^{N \times \dms}$, with embedding dimension $\dms$. With $\G$'s adjacency matrix $\A\in \{0,1\}^{N\times N}$, where $\A[u_i,v_j]=\A[v_j,u_i]=1$ if $e_{ij}\in \E$ and $\A[u_i,v_j]=0$ otherwise, the graph convolution process represents nodes as $\Z = \h(\A, \X)$, where $\h$ is a given graph convolution function. \eat{The affinity score between user $u_i$ and item $v_j$ is defined as the inner-product of their representation vectors, \ie $\afff = \langle \Z[u_i], \Z[v_j] \rangle$.} They frequently apply ranking losses for the optimization, \eg Bayesian Personalized Ranking (BPR)~\cite{rendle2012bpr}, which optimizes each user $u_i$ to have a higher inner product between $\Z[u_i]$ and $\Z[v_j]$, where $v_j$ is the item purchased by $u_i$. It leads to a user having a similar representation vector with her neighbor users who share many identical purchased products, as well as their other purchased products. In this way, her neighbors' other purchases will also have a larger chance of being recommended to this user.

\stitle{Recommend with personas}\label{sec:lgcn3}
In addition to the observed purchases in $\G$, persona-based representations provide external information for each user. Leveraging LLM's real-world knowledge and reasoning abilities, these additional features include deep insights. Associating personas with their highly interested products allows for a more comprehensive consideration of the candidate items. For example, items with sparse purchases can attract sufficient attention if they align well with certain persona combinations. On the other hand, users with minor personas will avoid overemphasizing the purchases of mainstream users and popular items. Here we propose a straightforward method to integrate persona-based representations with graph convolution-based algorithms. First, we transform the original bipartite graph $\G$ to a new tripartite graph $\Gtri=(\U,\V,\PP,\Etri)$ by adding persona nodes as a new partition. In $\Gtri$, the edge set $\Etri$ contains the original edges in $\E$ and the edges between persona nodes and the associated user and item nodes. Specifically, the edge between persona and user nodes can be constructed using the proposed representation $\Pa$. Furthermore, the edges between persona and item nodes can be determined by querying an LLM to identify which personas an item pertains to, leveraging the LLM's reasoning capabilities and the description of each persona. We extend trainable embeddings to include persona nodes as $\Xtri\in \mathbb{R}^{\Ntri \times \dms}$ and reconstruct the adjacency matrix based on $\Gtri$. Without changing the graph convolution function $\h$, the representation vectors on the tripartite graph are given by $\Ztri$. \eat{During training, we only consider the ranking loss between user nodes and item nodes. The optimized embeddings can be defined as
$\hat{\Xtri} = \arg \min_{\Xtri} \ Loss\left(\h(\Atri, \Xtri), \E \right).$}
Except for adding an external set of trainable embeddings for persona nodes, only the propagation procedure changes due to modifying the adjacency matrix. This facilitates the easy migration of mainstream graph convolution-based algorithms, such as LGCN~\cite{yu2022low} and AFDGCF~\cite{wu2024afdgcf}, from the bipartite graph to the tripartite graph with personas by simply replacing the $\A$ with $\Atri$, even without explicit modification of the implementation.
% Section: Experiments
\begin{table}[!t]
\caption{Statistics of datasets.}
\vspace{-2mm}
\begin{small}
\begin{tabular}{lllll}
\hline
Dataset   & User\#   & Item\#   & Interaction\# & Sparsity  \\ \hline
\mba      & 4,297  & 3,846  & 263,267      & 98.4070\% \\
\ins      & 20,620 & 41,521 & 1,333,805    & 99.8442\% \\
\insfull  & 206,209 & 49,677 & 13,307,953 & 99.8701\% \\ \hline
\end{tabular}
\end{small}
\label{tbl:datasets}
\vspace{-2mm}
\end{table}

\setlength{\tabcolsep}{2pt}
\begin{table*}[!t]
\centering
\begin{small}
\caption{Performance evaluation in NDCG@K (N@K) and F1-Score@K (F@K).}\vspace{-3mm}
\label{tbl:compare-merge}
\resizebox{\textwidth}{!}{
\begin{tabular}{c|cccccc|cccccc|ccccc}
\hline
\multirow{2}{*}{} &
  \multicolumn{6}{c|}{OnelineRetail} &
  \multicolumn{6}{c|}{Instacart} &
  \multicolumn{5}{c}{Instacart Full} \\ \cline{2-18} 
 &
  MF &
  Light &
  LGCN &
  \multicolumn{1}{c|}{AFD} &
  LGCN3 &
  A-LGCN3 &
  MF &
  Light &
  LGCN &
  \multicolumn{1}{c|}{AFD} &
  LGCN3 &
  A-LGCN3 &
  MF &
  Light/AFD &
  \multicolumn{1}{c|}{LGCN} &
  LGCN3 &
  A-LGCN3 \\ \hline
N@2 &
  0.2391 &
  0.2801 &
  0.2686 &
  \multicolumn{1}{c|}{0.2898} &
  0.2933 &
  \textbf{0.2940} &
  0.1166 &
  0.1477 &
  0.1405 &
  \multicolumn{1}{c|}{0.1535} &
  0.1570 &
  \textbf{0.1634} &
  0.1149 &
  OOM &
  \multicolumn{1}{c|}{0.1485} &
  \textbf{0.1540} &
  0.1533 \\
N@5 &
  0.2143 &
  0.2497 &
  0.2356 &
  \multicolumn{1}{c|}{0.2578} &
  \textbf{0.2602} &
  0.2549 &
  0.1006 &
  0.1273 &
  0.1208 &
  \multicolumn{1}{c|}{0.1308} &
  0.1319 &
  \textbf{0.1357} &
  0.0972 &
  OOM &
  \multicolumn{1}{c|}{0.1248} &
  \textbf{0.1300} &
  0.1296 \\
N@10 &
  0.2104 &
  0.2443 &
  0.2274 &
  \multicolumn{1}{c|}{0.2489} &
  \textbf{0.2503} &
  0.2497 &
  0.0916 &
  0.1180 &
  0.1106 &
  \multicolumn{1}{c|}{0.1205} &
  0.1197 &
  \textbf{0.1230} &
  0.0879 &
  OOM &
  \multicolumn{1}{c|}{0.1126} &
  \textbf{0.1187} &
  0.1182 \\
N@20 &
  0.2221 &
  0.2575 &
  0.2383 &
  \multicolumn{1}{c|}{0.2595} &
  0.2577 &
  \textbf{0.2617} &
  0.0934 &
  0.1198 &
  0.1117 &
  \multicolumn{1}{c|}{0.1222} &
  0.1206 &
  \textbf{0.1242} &
  0.0879 &
  OOM &
  \multicolumn{1}{c|}{0.1130} &
  \textbf{0.1210} &
  0.1200 \\
N@50 &
  0.2612 &
  0.2942 &
  0.2772 &
  \multicolumn{1}{c|}{0.2978} &
  0.2940 &
  \textbf{0.2996} &
  0.1150 &
  0.1448 &
  0.1351 &
  \multicolumn{1}{c|}{0.1476} &
  0.1451 &
  \textbf{0.1491} &
  0.1052 &
  OOM &
  \multicolumn{1}{c|}{0.1363} &
  \textbf{0.1474} &
  0.1463 \\
N@100 &
  0.3011 &
  0.3337 &
  0.3167 &
  \multicolumn{1}{c|}{0.3346} &
  0.3325 &
  \textbf{0.3392} &
  0.1389 &
  0.1733 &
  0.1618 &
  \multicolumn{1}{c|}{0.1761} &
  0.1732 &
  \textbf{0.1766} &
  0.1268 &
  OOM &
  \multicolumn{1}{c|}{0.1634} &
  \textbf{0.1773} &
  0.1758 \\ \hline
F@2 &
  0.0859 &
  0.1080 &
  0.0955 &
  \multicolumn{1}{c|}{0.1066} &
  0.1052 &
  \textbf{0.1092} &
  0.0326 &
  0.0415 &
  0.0375 &
  \multicolumn{1}{c|}{0.0427} &
  0.0411 &
  \textbf{0.0433} &
  0.0310 &
  OOM &
  \multicolumn{1}{c|}{0.0376} &
  \textbf{0.0396} &
  0.0392 \\
F@5 &
  0.1123 &
  0.1333 &
  0.1207 &
  \multicolumn{1}{c|}{0.1348} &
  0.1332 &
  \textbf{0.1349} &
  0.0478 &
  0.0609 &
  0.0564 &
  \multicolumn{1}{c|}{\textbf{0.0629}} &
  0.0611 &
  0.0628 &
  0.0457 &
  OOM &
  \multicolumn{1}{c|}{0.0570} &
  \textbf{0.0604} &
  0.0599 \\
F@10 &
  0.1225 &
  0.1416 &
  0.1352 &
  \multicolumn{1}{c|}{\textbf{0.1436}} &
  0.1431 &
  0.1432 &
  0.0560 &
  0.0716 &
  0.0672 &
  \multicolumn{1}{c|}{0.0730} &
  0.0722 &
  \textbf{0.0738} &
  0.0539 &
  OOM &
  \multicolumn{1}{c|}{0.0684} &
  \textbf{0.0727} &
  0.0722 \\
F@20 &
  0.1239 &
  0.1398 &
  0.1341 &
  \multicolumn{1}{c|}{0.1393} &
  0.1406 &
  \textbf{0.1415} &
  0.0602 &
  0.0752 &
  0.0710 &
  \multicolumn{1}{c|}{0.0768} &
  0.0758 &
  \textbf{0.0780} &
  0.0571 &
  OOM &
  \multicolumn{1}{c|}{0.0726} &
  \textbf{0.0779} &
  0.0775 \\
F@50 &
  0.1081 &
  0.1181 &
  0.1159 &
  \multicolumn{1}{c|}{0.1186} &
  0.1201 &
  \textbf{0.1210} &
  0.0565 &
  0.0684 &
  0.0644 &
  \multicolumn{1}{c|}{0.0692} &
  0.0681 &
  \textbf{0.0701} &
  0.0525 &
  OOM &
  \multicolumn{1}{c|}{0.0661} &
  \textbf{0.0713} &
  0.0709 \\
F@100 &
  0.0886 &
  0.0955 &
  0.0934 &
  \multicolumn{1}{c|}{0.0945} &
  0.0967 &
  \textbf{0.0976} &
  0.0470 &
  0.0562 &
  0.0528 &
  \multicolumn{1}{c|}{0.0568} &
  0.0557 &
  \textbf{0.0570} &
  0.0434 &
  OOM &
  \multicolumn{1}{c|}{0.0545} &
  \textbf{0.0587} &
  0.0583 \\ \hline
\end{tabular}
}
\end{small}
\end{table*}
\section{Experiments}
In this section, we evaluate our proposed \ours, \yimin{mainly} focusing on its application to personalized recommendations on real-world e-commerce datasets. 
We aim at the following five research questions:
\begin{itemize}[topsep=2pt,itemsep=1pt,parsep=0pt,partopsep=0pt,leftmargin=11pt]
\item \textbf{RQ1}: How do persona-based representations improve the effectiveness of the personalized recommendation?
\item {\textbf{RQ2}: How does the choice of LLM, sampling budget, and random walk length affect \ours's performance \yimin{in the recommendation?}}
\yimin{
\item \textbf{RQ3}: How do persona-based representations outperform existing explicit customer representations in customer segmentation?}
\item \textbf{RQ4}: How much speedup does the reverse approximate solution yield in affinity computation?
\end{itemize}
{All experiments are conducted on a Linux machine with 12 Intel(R) Xeon(R) CPU @ 2.20GHz, 52GB of RAM and an NVIDIA L4 (24GB) GPU.}

\subsection{Experimental Settings}
\stitle{Dataset description}
To evaluate how integrating customer representations by personas improves the e-commerce personalized recommendation, we conduct experiments on the following publicly accessible real-world datasets that vary in size and sparsity:

\mba, containing sales data from a European retailer in 2010~\cite{ds_mba};
\ins, consisting of randomly sampled 20,620 anonymized users and their purchase data grocery orders in 2017~\cite{ds_ins}; \insfull, the full version of \ins dataset.
The statistics of these datasets are summarized in Table \ref{tbl:datasets}. For each dataset, we randomly select approximately $80\%$ of each user's historical purchases as the training set to construct $\G$, leaving the remaining $20\%$ as the test set.

\stitle{Model configurations}
\yimin{
To evaluate the performance of persona-based representation in the recommendation, we select two state-of-the-art models, LGCN~\cite{yu2022low} and AFDGCF~\cite{wu2024afdgcf}. 
% story
Specifically, LGCN is an efficient and scalable graph convolution-based (GCN) model that is capable of processing large graphs with more than tens of millions of interaction edges.
AFDGCF improves the recommendation accuracy by adding a de-correlation loss term to GCN models, including LightGCN~\cite{he2020lightgcn} and LGCN.
We improve LGCN by integrating personas using the procedure explained in Section~\ref{sec:lgcn3} and name the enhanced model \lgcntri.
Due to the scalability issue of the default LightGCN model in AFDGCF, we choose to further improve \lgcntri by adding AFDGCF's de-correlation loss and named this model \afdtri.
To ensure a fair comparison, we use the same embedding size and fine-tuned the hyperparameters for all improved models and compared baselines respectively.}

For \ours, we set with the cutoff constant to $5$, the number of iterations as $10$ and employ OpenAI's GPT-4 API~\cite{achiam2023gpt} %(gpt-4-0125-preview) 
with its default settings to generate $\Pa$ and conduct \LLM. For computing the user-persona affinity matrix, we set the random walk length $\stp=1$ and de-bias hyperparameter $\beta=0.5$. In addition, we use \rrw with $\err=0.001$ for \insfull while using the exact solution for other datasets.
All implementation details including datasets, prompt templates, and algorithms are available at: \color{blue}
\url{https://github.com/Hanc1999/YouAreWhatYouBought}.
\color{black}

\stitle{Evaluation metric}
We employ NDCG@K and F1-Score@K as evaluation metrics on the recommendation results with K from $\{2,5,10,20,50,100\}$ to ensure a comprehensive comparison, where $\text{NDCG@K} = \frac{\text{DCG@K}}{\text{IDCG@K}}$ and $\text{F1-Score@K} = 2 \times \frac{\text{Precision@K} \times \text{Recall@K}}{\text{Precision@K} + \text{Recall@K}}$. We repeat each experiment three times and report the average performance.

\subsection{Recommendation Performance (RQ1)} 
\yimin{
\stitle{Overall evaluations}
% baselines
We compare the performance of \lgcntri and \afdtri with six well-adopted baselines: MF~\cite{koren2009mf}, GCMC~\cite{berg2017graph}, LCFN~\cite{yu2020graph}, LightGCN~\cite{he2020lightgcn}, LGCN~\cite{yu2022low} and AFDGCF~\cite{wu2024afdgcf}.
Table~\ref{tbl:compare-merge} presents the main performance results of MF, LightGCN (Light), LGCN, AFDGCF (AFD), \lgcntri and \afdtri in terms of NDCG@K and F1-Score@K, with the best method highlighted in bold. Due to the space constraints, we omit LCFN and GCMC, which perform lower on most metrics than LightGCN.
% performance
As demonstrated in Table~\ref{tbl:compare-merge}, we observe that after integrating LGCN with persona-based customer representations generated by \ours, \lgcntri significantly improves NDCG@K and F1-Score@K metrics by up to 10.4\%, 11.7\%, and 8.5\% across three real-world datasets. Meanwhile, \lgcntri maintains LGCN's efficiency and scalability.
% response to the t-test requirement
We further conduct paired t-tests between LGCN and \lgcntri across all three datasets. The average p-values for F1-Score@K and NDCG@K are 0.027 and 0.012, suggesting the improvement is statisticatlly meaningful.
After introducing the same de-correlation loss term from the AFDGCF framework to \lgcntri, the recommendation performance of \afdtri is further improved on \mba and \ins datasets. It outperforms AFDGCF in most metrics by up to 6.4 \% without affecting its efficiency.
In addition, when processing the largest dataset \insfull, LGCN, \lgcntri, and \afdtri demonstrate better scalability, whereas LCFN, LightGCN, and AFDGCF exceed the RAM limit.
}

\stitle{Case study}
In Table \ref{tbl:case2}, the first column lists the test set products for User 156246 (unordered). The following columns display the top-10 recommendations returned by LGCN and \lgcntri, ranked in descending order based on the inner product of the corresponding user and item embedding in LGCN. A recommended item is bolded if it appears in the test set. This user's 39 historical purchases in the training set include nine organic foods, six fruits, six vegetables, four baby foods, and some healthy items like yogurts and chicken breasts. Accordingly, \ours assigned her with personas \textit{Organic Foodie}, \textit{Health Enthusiast}, and \textit{Baby Care Provider}. 
LGCN, without persona considerations, mostly recommended popular fruits and vegetables, {with all recommendations being among the top 100 most popular out of 41,521 products, including six in the top 20.} However, it ignores the user's preference for baby foods, resulting in only one product (Banana) overlapping with the test set, which contains six other baby foods.
{In contrast,} with the persona-based representation, \lgcntri prioritizes the baby food category, which is much less popular than categories like fruits and vegetables in this dataset due to data bias. It recommends two relatively popular items in the baby food category and puts them in high-ranking positions, successfully hitting two additional products in the test set. Furthermore, \lgcntri removed three top-20 popular false positives from LGCN's results, indicated by the under-wave. {In this way, personas help \lgcntri to reduce overemphasis on highly popular items, enabling fairer consideration of other products that align with the user’s personas.}

\stitle{Personas vs. categorical features} 
\yimin{
To demonstrate the effectiveness of personas compared to traditional categorical features, we introduce LGCNL as another extension of LGCN, which replaces personas with category-based representations~\cite{li2015mining} and follows the procedure in Section~\ref{sec:lgcn3}. Due to the lack of category information in \mba, we report LGCNL’s performance on \ins. As shown in Table~\ref{tbl:lifestyle}, \lgcntri consistently outperforms LGCNL across all NDCG@K and F1-Score@K metrics. For instance, \lgcntri improves upon LGCNL by up to 3.6\% in NDCG@K and 2.0\% in F1-Score@K. This demonstrates that our generated personas are more informative and effective in describing a customer's purchase preference, leading to more improvement in the recommendation.
}

\begin{table}
\centering
\begin{small}
\caption{Case study for User 156246.}\vspace{-3mm}
\label{tbl:case2}
\resizebox{\columnwidth}{!}{%
\begin{tabular}{c|c|c|c}
\hline
\textbf{Test Set Products} & \textbf{\#} & \textbf{LGCN} & \textbf{\lgcntri} \\ 
\hline
Babyfood (Broccoli) & 1 & \textbf{\textit{Banana}} & \textbf{\textit{Banana}} \\ 
Pasta Sauce & 2 & Strawberries & Strawberries \\ 
Babyfood (Pumpkin) & 3 & Limes & Limes \\ 
Cheese Slices & 4 & Red Onion & \textbf{Babyfood (Spinach)} \\ 
Babyfood (Mighty) & 5 & \uwave{Yellow Onions} & \textbf{Babyfood (Carrot)} \\ 
\textbf{\textit{Banana}} & 6 & Bunched Cilantro & Spinach \\ 
\textbf{Babyfood (Spinach)} & 7 & \uwave{Organic Strawberries} & Red Onion \\ 
\textbf{Babyfood (Carrot)} & 8 & Whole Wheat Bread & Bunched Cilantro \\ 
Babyfood (Beet) & 9 & \uwave{Cucumber Kirby} & Red Peppers \\ 
Organic Cucumbers & 10 & Green Onions & Organic Fuji Apple \\ 
\hline
\end{tabular}
}
\end{small}
\end{table}

\setlength{\tabcolsep}{10pt}
\begin{table}[!t]
\centering
\begin{small}
\caption{LGCN3 with different personas on \ins.}\vspace{-2mm}
\label{tbl:lifestyle}

\begin{tabular}{c|cc|cc}
\hline
\multirow{2}{*}{K} & \multicolumn{2}{c|}{NDCG@K} & \multicolumn{2}{c}{F1-Score@K} \\ \cline{2-5} 
                   & LGCN3             & LGCNL   & LGCN3              & LGCNL     \\ \hline
2                  & \textbf{0.1570}   & 0.1516  & \textbf{0.0411}    & 0.0403    \\
5                  & \textbf{0.1319}   & 0.1282  & \textbf{0.0611}    & 0.0602    \\
10                 & \textbf{0.1197}   & 0.1173  & \textbf{0.0722}    & 0.0713    \\
20                 & \textbf{0.1206}   & 0.1191  & \textbf{0.0758}    & 0.0749    \\
50                 & \textbf{0.1451}   & 0.1438  & \textbf{0.0681}    & 0.0673    \\
100                & \textbf{0.1732}   & 0.1707  & \textbf{0.0557}    & 0.0549    \\ \hline
\end{tabular}
\vspace{-2mm}
\end{small}
\end{table}

\subsection{Ablation Study (RQ2)}

\setlength{\tabcolsep}{3pt}
\begin{table}[!t]
\centering
\begin{small}
\caption{LGCN3 with various LLMs on \mba.}\vspace{-2mm}
\label{tbl:compare-llm}

\begin{tabular}{c|cc|cc}
\hline
\multirow{2}{*}{K} & \multicolumn{2}{c|}{NDCG@K}       & \multicolumn{2}{c}{F1-Score@K} \\ \cline{2-5} 
                   & Llama-3-70B     & GPT-4-Turbo     & Llama-3-70B      & GPT-4-Turbo \\ \hline
2                 & 0.2925          & \textbf{0.2933} & 0.1050  & \textbf{0.1052}      \\
5                 & 0.2553  & \textbf{0.2601}         & \textbf{0.1355}  & 0.1332      \\
10                 & 0.2455          & \textbf{0.2503} & \textbf{0.1440}  & 0.1431      \\
20                 & \textbf{0.2578} & 0.2578          & \textbf{0.1428}  & 0.1406      \\
50                 & \textbf{0.2965} & 0.2940          & \textbf{0.1217}  & 0.1201      \\
100                & \textbf{0.3330} & 0.3325          & \textbf{0.0972}  & 0.0967      \\ \hline
\end{tabular}
\vspace{-2mm}
\end{small}
\end{table}
\begin{figure}[t]
\centering
\begin{small}
\begin{tikzpicture}
    \begin{customlegend}[legend columns=5,
        legend entries={K=2,K=5,K=10,K=20,K=50,K=100},
        legend columns=-1,
        area legend,
        legend style={at={(0.45,1.15)},anchor=north,draw=none,font=\scriptsize,column sep=0.15cm}]
        \addlegendimage{color=black,fill=red} %K=2
        \addlegendimage{color=black,fill=pink} %K=5
        \addlegendimage{color=black,fill=orange} %K=10
        \addlegendimage{color=black,pattern color=violet,pattern=dots} %K=20
        \addlegendimage{color=black,pattern color=teal,pattern=north west lines} %K=50
        \addlegendimage{color=black,pattern color=blue,pattern=crosshatch} %K=100
        % \addlegendimage{color=black,pattern color=teal,pattern=north west lines} %K=100
    \end{customlegend}
\end{tikzpicture}
\\[-\lineskip]
\hspace{-2mm}
\subfloat{
\hspace{-3.5mm}
\hbox{
\begin{tikzpicture}[scale=1]
\begin{axis}[
    height=\columnwidth/2.4,
    width=\columnwidth/1.6,
    ybar=1.0pt,
    bar width=0.08cm,
    enlarge x limits=true,
    ylabel={\em NDCG},
    xlabel= {\em sample rate},
    xmin=0.8, xmax=4.2,
    xtick={1,2,3,4},
    xticklabels={5\%,10\%,20\%,100\%},
    xticklabel style = {font=\footnotesize},
    ymin=0,
    ymax=0.4,
    ytick={0,0.1,0.2,0.3},
    yticklabels={0,0.1,0.2,0.3},
    % ymode=log,
    yticklabel style = {font=\footnotesize},
    % log basis y={10},
    every axis y label/.style={at={(current axis.north west)},right=50mm,above=0mm},
    legend style={at={(0.02,0.98)},anchor=north west,cells={anchor=west},font=\tiny}
    ]
\addplot [color=black,fill=red] coordinates {
(1,	0.2967)
(2,	0.2961)
(3,	0.2899)
(4,	0.2933)
};%K=2
\addplot [color=black,fill=pink] coordinates {
(1,	0.2609)
(2,	0.2559)
(3,	0.2557)
(4,	0.2601)
};%K=5
\addplot [color=black,fill=orange] coordinates {
(1,	0.2510)
(2,	0.2474)
(3,	0.2512)
(4,	0.2503)
};%K=10
\addplot [color=black,pattern color=violet,pattern=dots] coordinates {
(1,	0.2589)
(2,	0.2563)
(3,	0.2604)
(4,	0.2578)
};%K=20
\addplot [color=black,pattern color=teal,pattern=north west lines] coordinates {
(1,	0.2944)
(2,	0.2960)
(3,	0.2988)
(4,	0.2940)
};%K=50
\addplot [color=black,pattern color=blue,pattern=crosshatch] coordinates {
(1,	0.3333)
(2,	0.3344)
(3,	0.3350)
(4,	0.3325)
};%K=100
\end{axis}
\end{tikzpicture}\vspace{2mm}\hspace{-1mm}%
\hspace{1mm}
\begin{tikzpicture}[scale=1]
\begin{axis}[
    height=\columnwidth/2.4,
    width=\columnwidth/1.6,
    ybar=1.0pt,
    bar width=0.08cm,
    enlarge x limits=true,
    ylabel={\em F1-Score},
    xlabel= {\em sample rate},
    xmin=0.8, xmax=4.2,
    xtick={1,2,3,4},
    xticklabels={5\%, 10\%, 20\%, 100\%},
    xticklabel style = {font=\footnotesize},
    ymin=0,
    ymax=0.2,
    ytick={0, 0.05, 0.1, 0.15},
    yticklabels={0,0.05,0.1,0.15},
    % ymode=log,
    yticklabel style = {font=\footnotesize},
    % log basis y={10},
    every axis y label/.style={at={(current axis.north west)},right=5mm,above=0mm},
    legend style={at={(0.02,0.98)},anchor=north west,cells={anchor=west},font=\tiny}
    ]
    % times 100 for better viz
\addplot [color=black,fill=red] coordinates {
(1,	0.1062)
(2,	0.1062)
(3,	0.1062)
(4,	0.1052)
};%K=2
\addplot [color=black,fill=pink] coordinates {
(1,	0.1336)
(2,	0.1346)
(3,	0.1341)
(4,	0.1332)
};%K=5
\addplot [color=black,fill=orange] coordinates {
(1,	0.1460)
(2,	0.1451)
(3,	0.1436)
(4,	0.1431)
};%K=10
\addplot [color=black,pattern color=violet,pattern=dots] coordinates {
(1,	0.1436)
(2,	0.1442)
(3,	0.1412)
(4,	0.1406)
};%K=20
\addplot [color=black,pattern color=teal,pattern=north west lines] coordinates {
(1,	0.1217)
(2,	0.1214)
(3,	0.1207)
(4,	0.1201)
};%K=50
\addplot [color=black,pattern color=blue,pattern=crosshatch] coordinates {
(1,	0.0966)
(2,	0.0965)
(3,	0.0966)
(4,	0.0967)
};%K=100
\end{axis}
\end{tikzpicture}\vspace{2mm}\hspace{-1mm}%
} % hbox
} % subfloat
\vspace{-1mm}%}
\end{small}
\vspace{-3mm}
\caption{\lgcntri with different sample rates on \mba.} \label{fig:samplerate}
\vspace{-2mm}
\end{figure}

\stitle{LLMs}
\yimin{To explore the performance of \lgcntri with different LLMs, we replace the default GPT-4-Turbo model with Llama-3-70B~\cite{siriwardhana2024domain} in \ours and report the results on \mba in Table~\ref{tbl:compare-llm}.
As shown, there is no significant difference in their performance on NDCG@K and F1-Score@K. Specifically, \lgcntri with GPT-4-Turbo improves the NDCG@K and F1-Score@K on the \mba dataset by up to 10.4\% and 10.3\%. \lgcntri with Llama-3-70B improves these metrics by up to 8.9\% and 12.2\%. This indicates that the quality of persona representations returned by \ours is not sensitive to the LLMs.
}

\stitle{Sampling budget}
{To evaluate \lgcntri with different LLM budgets $\tau$, we run \ours by setting $\tau\in\{5\%, 10\%, 20\%, 100\%\}\times |\U|$, As reported in Figure~\ref{fig:samplerate}, there are no significant differences in recommendation performance across sample rates for either metric. Even with a sample rate of $5\%$, the results on both datasets nearly match the performance achieved when all persona representations are generated by the LLM. This indicates that our random walk-based method effectively infers personas for unlabeled users, reducing LLM usage costs without sacrificing downstream performance.
Consistent patterns emerge across all sample rates for both metrics on \mba. The F1-Score@K starts low, rises with $K$ up to around 10–20, then declines as recall improves at low $K$ but precision drops at higher values. In contrast, NDCG@K dips near $K=10$ before increasing, as iDCG stabilizes after initially rising to match the average test set size, allowing DCG growth to drive NDCG@K upward.}

\yimin{
\stitle{Random walk length}
To evaluate the impact of $\Stp$, we choose LGCN3 with a 20\% sampling rate and set $\Stp$ to 1 and 2, respectively. We then evaluate their performance on NDCG@K and F1-Score@K on the \mba dataset. As reported in Table~\ref{tbl:compare-stp}, while both settings show advantages across different metrics, the average improvement over LGCN for $\Stp=1$ is 8.3\%, which is higher than the 7.3\% for $\Stp=2$.
This might be because the attention score becomes less discernible as the random walk length increases.
\setlength{\tabcolsep}{3pt}
\begin{table}[!t]
\centering
\begin{small}
\caption{LGCN3 with various $\Stp$ on \mba.}\vspace{-2mm}
\label{tbl:compare-stp}

\begin{tabular}{c|cc|cc}
\hline
\multirow{2}{*}{K} & \multicolumn{2}{c|}{NDCG@K}                     & \multicolumn{2}{c}{F1-Score@K}                  \\ \cline{2-5} 
                   & LGCN3 ($\hat{\ell}=1$) & LGCN3 ($\hat{\ell}=2$) & LGCN3 ($\hat{\ell}=1$) & LGCN3 ($\hat{\ell}=2$) \\ \hline
2  & 0.2899 & \textbf{0.2938} & \textbf{0.1062} & 0.1059         \\
5  & 0.2557 & \textbf{0.2562} & \textbf{0.1341} & 0.1333          \\
10  & \textbf{0.2511} & 0.2444 & \textbf{0.1436} & 0.1425          \\
20  & \textbf{0.2604} & 0.2557 & 0.1412          & \textbf{0.1416} \\
50  & \textbf{0.2988} & 0.2941 & 0.1207          & \textbf{0.1213} \\
100 & \textbf{0.3349} & 0.3327 & 0.0966          & \textbf{0.0968} \\ \hline
\end{tabular}
\vspace{-4mm}
\end{small}
\end{table}
}

\subsection{Customer Segmentation (RQ3)}\label{app:d}
\yimin{
In this set of experiments, we evaluate the quality of personas generated by \ours against existing explicit customer representations in customer segmentation, which is another e-commerce application introduced in Section~\ref{sec:app}.
Specifically, the quality is measured by (i) their robustness over time and (ii) their clustering quality.
}

\stitle{Robustness}
\yimin{
To evaluate the robustness over time, we first randomly sample 300 customers from the \mba dataset with more than 10 transactions over a year and then divide their transaction data into the first and second six-month periods.
We include the well-adopted RFM model~\cite{wan2022fast} in customer segmentation, which calculates users' Recency, Frequency, and Monetary values as their representations.
For persona representation, we apply the LLM (GPT-4-Turbo) to select the three most dominant personas for each user from the previously generated 20-persona set. To ensure a fair comparison, each dimension of the RFM representation is discretized into 10 quantile-based baskets.
We call a user a consistent customer if their representation does not change between the two periods. We use the number and fraction of consistent customers to measure the representation's robustness over time, with the results shown in Table~\ref{tbl:robust}. We can observe that the fraction of consistent customers for persona representation is $13.8\times$ higher than the RFM model, demonstrating its superior robustness.
}

\stitle{Cluster quality}
\yimin{
To demonstrate the effectiveness of the persona in customer segmentation, we compare it with the competing RFM model. Specifically, we encode the persona representation using one-hot encoding and reduce it to three dimensions via PCA~\cite{mackiewicz1993principal} to align with the RFM, with both representations L2-normalized.
Using the user representations from the first six months of data in \mba, we perform K-means clustering with varying numbers of clusters \{5, 15, 25, 35\} and evaluate clustering quality using the Silhouette Score~\cite{rousseeuw1987silhouettes}. As demonstrated in Table~\ref{tbl:c_segmentation}, the proposed persona representation outperforms the state-of-the-art competitor in the customer segmentation performance by an average of 61.3\%.
}

\subsection{Approximate Solution Evaluation (RQ4)}\label{app:b}

\begin{table}[!t]
    \centering
    \small
    \caption{Robustness on \mba.}\vspace{-4mm}
    \begin{tabular}{c|c}
        \hline
        \textbf{Method} & \textbf{\# of Consistent Customers} \\ \hline
        RFM           & 4 (1.3\%)                          \\ \hline
        Persona                        & 54 (18\%)                          \\ \hline
    \end{tabular}
    \label{tbl:robust}
\end{table}
\begin{table}[!t]
    \centering
    \small
    \caption{Silhouette score on \mba with varying Clusters\# (Larger is better).}\vspace{-3mm}
    \begin{tabular}{c|cccc}
    \hline
    \textbf{Clusters\#}     & 5              & 15             & 25             & 35             \\  \hline
    \textbf{RFM}     & 0.366          & 0.404          & 0.431          & 0.445          \\ \hline
    \textbf{Persona} & \textbf{0.451} & \textbf{0.671}  & \textbf{0.771} & \textbf{0.788} \\ \hline
    \end{tabular}
    \label{tbl:c_segmentation}
\end{table}
\begin{table}
    \centering
    \small
    % \caption{Time cost and average absolute error of \rrw with different $\err$ on \insfull dataset.}\vspace{-3mm}
    \caption{\rrw on \insfull with varying $\err$.}\vspace{-3mm}
    \label{tbl:time-error-tradeoff}
     \begin{tabular}{c|cccc}
    \hline
    \textbf{$\epsilon$} & 0.0   & 0.02 & 0.05  & 0.1  \\ \hline
    \textbf{AAE (e-3)}  & 0.0   & 0.92 & 1.59  & 1.73 \\ \hline
    \textbf{Time (s)}   & 268.8 & 77.7 & 49.79 & 47.86 \\ \hline
\end{tabular}
\vspace{-2mm}
\end{table}
For large-scale e-commerce datasets such as \insfull, the original exact solution based on the matrix computation in Eq.~\eqref{eq:exactaff} is impractical. The intermediate product $\hat{\mathbf{P}}=\Qm \cdot \Qmm$ can easily exceed the RAM limits. For instance, with over 200 thousand users in the \insfull dataset, the constructed $\hat{\mathbf{P}}$ contains more than $40$ billion floating point numbers, which need $160$GB storage. To evaluate the empirical efficiency and accuracy of the proposed \rrw method, we compute the user-persona affinities with various error tolerance settings ranging from $\err=0$ to $\err=0.2$. We execute \rrw five times for each setting and record the average running time. The average absolute error (AAE) between each user's approximated user-persona affinities and the exact solution is calculated to represent the empirical error, defined as: 
$AAE(u_i)=\frac{1}{|\PP|}\sum_{\prs_m}|\Aff[u_i,\prs_m] - \hat{\Aff}[u_i, \prs_m]|$. To simplify the comparison, we remove the de-bias coefficients in Eq.~\eqref{eq:aff}. Table~\ref{tbl:time-error-tradeoff} reports the results, showing that empirical error is significantly smaller than the theoretical error tolerance $\err$. For example, with $\err=0.02$, the empirical AAE is only $0.00092$, about 20 times smaller. Regarding the time cost, computing exact solutions with \rrw ($\err=0$) takes around $268$ seconds. While approximating the affinities with an empirical AAE smaller than $2\times 10^{-3}$ takes about $51$ seconds, which is more than five times faster. This demonstrates the efficiency of our proposed \rrw method on large-scale datasets.

\section{Conclusion}
In this work, we introduce the concept of customer personas and propose an explicit customer representation through these personas. We then present \ours, an effective and efficient method for generating customer personas for e-commerce applications.
By leveraging random walks to infer customer personas from a small prototype user set, our approach significantly reduces the cost of full LLM labeling. Additionally, the proposed \rrw algorithm enables fast user-persona affinity computations, easily handling large-scale datasets in seconds. Integrating persona representations with state-of-the-art recommendation models, we demonstrate the superiority of our solution, which is also shown in customer segmentation.
For future work, we plan to explore the effectiveness of our method in further e-commerce applications and a wider range of user interaction scenarios.

\begin{acks}
This research is supported by the Ministry of Education, Singapore, under its MOE AcRF TIER 3 Grant (MOE-MOET32022-0001). 
\end{acks}

\bibliographystyle{ACM-Reference-Format}
\balance
\bibliography{ref}

\twocolumn[
\appendix
\section{Appendix}
]
\appendix

\section{Proof of Theorems}\label{app:a}
\subsection{Proof of Theorem~\ref{thm:1}}
\setcounter{theorem}{0}
\kdd{\begin{theorem}
% \label{thm:1}
The time complexity for computing the exact $\Aff$ by Eq.~\eqref{eq:aff} is $O(|E|\cdot|\U| + (\Stp-1)|\U|^3 + |\PP|\cdot|\U|^2)$.
% With $|E|$ represents the total number of edges in $G$.
\end{theorem}}
\begin{proof}
First, we need to compute $\Qm\cdot \Qmm$ in Eq.~\eqref{eq:exactaff} where both $\Qm$ and $\Qmm$ are sparse matrices. We have $\Deg_i = nnz(\Qm[u_i,:]) = nnz(\Qmm[:,u_i])$ and $\sum_{i\in \U} d_i = |E|$. We calculate the computational complexity of this step by summing up the computations needed for each of its output elements as follows:
\begin{align*}
    \sum_{i\in \U}\sum_{j\in \U} O(d_i + d_j) &=\sum_{i\in \U}O\left(|U|\cdot d_i + \sum_{j\in \U}d_j \right) \\
    &=\sum_{i\in \U}O(|\U|\cdot d_i + |E|)\\
    &=O\left(|\U|\cdot \sum_{i\in \U}d_i + |\U|\cdot |E|\right)\\
    &= O(|U|\cdot |E|)
\end{align*}
We now focus on the computation of $\Att_{\Stp}$, as it naturally provides the values for all $\Att_\stp$ with $\stp < \Stp$, as well as for $\Att$. Denote $\Qm \cdot \Qmm$ as $\hat{\mathbf{P}}$, the first step involves computing $\hat{\mathbf{P}}^\stp$. Since $\hat{\mathbf{P}}\in \mathbb{R}^{|\U|\times|\U|}$ is a dense matrix, the computational complexity of this step is $O((\Stp-1)|\U|^3)$. For the computation of user-persona affinities in Eq.~\eqref{eq:aff}, \kdd{we apply $\mathbf{\Pi} \cdot \Lbl$ where $\mathbf{\Pi}$ is also a dense matrix.} Assuming that each user's associated persona number is proportional to $|\PP|$, the $\Lbl$ is also dense. Consequently, the computational complexity of this step is $O(|\PP|\cdot |\U|^2)$. Summing up these complexities derives that the overall algorithm computational complexity is $O(|E|\cdot|\U| + (\stp-1)|\U|^3 + |\PP|\cdot|\U|^2)$. 
\end{proof}

\subsection{Proof of Theorem~\ref{thm:2}}
\setcounter{theorem}{2}
\begin{theorem}
    Given a persona $\prs_m$ and an error $\err$, the $\s[:,\prs_m]$ generated by \rrw is an $\err$-approximate user-persona affinity.
\end{theorem}
\begin{proof}
    We denote $s_t(\w, \prs_m) = \sum_{u_k\in \Ul} \tp_t(\w,u_k)\cdot \Lbl[u_k,\prs_m]$ 
    and $\Err_t=\max_{w\in\U\cup\V} |s_t(\w,\prs_m)-\s_t[\w,\prs_m]|$ with the maximizer $\w_t^*$.
    For $t=0$, by definition we have
    \begin{align*}
        \Err_0 &= |s_0(\w_0^*,\prs_m)-\s_0[\w_0^*,\prs_m]|\\
         &=|\tp_0(\w_0^*,\w_0^*)\cdot \Lbl[\w_0^*,\prs_m]-\Lbl[\w_0^*,\prs_m]|=0
    \end{align*}
    For $t=1$, we bound $\Err_1$ as the following:
    \begin{align*}
        \Err_1 &= |s_1(\w_1^*, \prs_m) - \s_1[w^*_1,\prs_m]|\\
        &= \frac{1}{|\nei{\w^*_1}|} \sum_{\w\in \nei{\w^*_1}}  |s_0(\w, \prs_m)-\s_0[\w,\prs_m]+\p_0(\w)|\\
        &\leq \frac{1}{|\nei{\w^*_1}|} \sum_{\w\in \nei{\w^*_1}}  |s_0(\w, \prs_m)-\s_0[\w,\prs_m]|+|\p_0(\w)|\\
        &\leq \frac{1}{|\nei{\w^*_1}|} \sum_{\w\in \nei{\w^*_1}} \left( \Err_0+\frac{\err}{2\Stp} \right) = \frac{\err}{2\Stp}
    \end{align*}
    Similarly, we have the following bound for $\Err_2$:
    \begin{align*}
        \Err_2 \leq \frac{1}{|\nei{\w_2^*}|} \sum_{\w\in \nei{\w_2^*}} \left( \Err_1+\frac{\err}{2\Stp}  \right) \leq \frac{2\err}{2\Stp}
    \end{align*}
    Inductively, we have the following bound for $\Err_{2\Stp}$:
    $$\Err_{2\Stp} \leq \frac{1}{|\nei{\w_{2\Stp}^*}|} \sum_{\w\in \nei{\w_{2\Stp}^*}} \left( \frac{{2\Stp}\err}{2\Stp} \right)=\err$$
    Finally, for any $u_i \in \U $ we have the following bound on the estimation error:
    % \begin{align*}
    %     |\Aff[u_i, \prs_m] - \s_{2\Stp}[u_i, \prs_m]|= |s_{2\Stp}(u_i, \prs_m)-\s_{2\Stp}[u_i, \prs_m]| \leq \Err_{2\Stp} \leq \err
    % \end{align*}
    \begin{align*}
        & \left|\Aff[u_i, \prs_m] - \frac{1}{\Stp} \sum\limits_{\stp\leq \Stp}\s_{2\stp}[u_i, \prs_m]\right| \\
        = & \left|\frac{1}{\Stp} \sum\limits_{\stp \leq \Stp} s_{2\stp}(u_i, \prs_m)- \frac{1}{\Stp} \sum\limits_{\stp\leq \Stp}\s_{2\stp}[u_i, \prs_m]\right| \\
        \leq & \frac{1}{\Stp} \sum\limits_{\stp \leq \Stp} \Err_{2\stp}\\
        \leq & \err
    \end{align*}
    
\end{proof}

\subsection{Proof of Theorem~\ref{thm:3}}
\setcounter{theorem}{3}
\begin{theorem}
    Given a graph $\G$, a persona set $\PP$, and an error $\err$, the time complexity of estimating $\Aff$ by invoking \rrw from every $r_m\in \PP$ is $O\left(\frac{1}{\err}\Stp^2\left(\NN\log(\Stp\NN)+|E|\right)\right)$, where $N=|U|+|V|$.
\end{theorem}
\begin{proof}
    Denote $c^* = \max_m c_m$.
    During the execution for a specific persona $\prs_m$, the maximal time to pop a node $w$ from \PQ for its $t$-th reverse updating is $\sum_{u_k\in\Ul}(\tp_t(\w,u_k)\cdot\Lbl[u_k,\prs_m])/(\err/{(2\Stp)})$. Because each time we process on $w$ will refresh its priority by at least $\err/{(2\Stp)}$, and its overall gained priority will not exceed the numerator term. In each process of $\w$, the time to update $\s_t[:,\prs_m]$ is $\Deg_{w}$ steps since each of its in-neighbors receives some of its updates. The maintenance of \PQ (poping and updating) each time costs $O(\log({2\Stp} \NN))$ work. Thus the running time for a single $w$ for updating all its hops is less than the following.
    \begin{equation*}
      \frac{1}{\err/{(2\Stp)}}\cdot \sum_t\sum_{u_k\in\Ul}(\tp_t(\w,u_k)\cdot \Lbl[w,\prs_m])\cdot O(\log({2\Stp}\NN)+\Deg_{\w}) 
    \end{equation*}
    Then for all personas and all nodes, the overall cost can be bounded as the following:
    \begin{align*}
        &\frac{1}{\err/{(2\Stp)}}\sum_{\prs_m\in \PP}\sum_{\w\in \U\cup\V} \sum_t\sum_{u_k\in\Ul}(\tp_t(\w,u_k)\cdot \Lbl[\w,\prs_m])\cdot O(\log({2\Stp}\NN)+\Deg_{\w})\\
        =&\frac{1}{\err/{(2\Stp)}}\sum_{\w\in\U\cup\V}\left(\sum_t\sum_{u_k\in\Ul}\tp_t(\w,u_k)\sum_{\prs_m\in\PP} \Lbl[\w,\prs_m]\right)\cdot O(\log({2\Stp}\NN)+\Deg_{\w})\\
        \leq&\frac{1}{\err/{(2\Stp)}}\sum_{\w\in\U\cup\V}2\Stp\cdot c^*\cdot O(\log({2\Stp}\NN)+\Deg_{\w})\\
        =&\frac{c^*4\Stp^2}{\err}\left(\sum_{\w\in\U\cup\V} O(\log({2\Stp}\NN)) + \sum_{\w\in\U\cup\V}O(\Deg_{\w})\right)\\
        =&\frac{c^*4\Stp^2}{\err} \left(\NN \cdot O(\log({2\Stp} \NN)) + O(|E|)\right)\\
        =& O\left(\frac{c^*\Stp^2}{\err}\left(\NN\log({2\Stp}\NN)+|E|\right)\right)\\
    \end{align*}
    Since $c^*\leq1$, we have the eventual computational complexity $O\left(\frac{1}{\err}\Stp^2\left(\NN\log(\Stp\NN)+|E|\right)\right)$.
\end{proof}

\renewcommand{\thefigure}{2.\alph{figure}}
\setcounter{figure}{0}

\begin{figure*}[b!]
% \noindent
\centering
\begin{tcolorbox}[
    enhanced,
    title=Market Basket Analysis (MBA),
    frame style image=blueshade.png,
    colback=yellow!10!white,
    fonttitle=\bfseries,
    overlay={\node[anchor=south east,font=\bfseries\itshape\huge,xshift=-5mm,yshift=3mm] at (frame.south east){Step 1};}
    ]

\textbf{System Prompt:} You are an assistant skilled at summarizing, capable of deducing high-level consumer keywords based on a user's purchases.

\textbf{User Prompt:} Take a deep breath and work according to the instructions step by step. Now you will conduct a series of analyses on the Market Basket Analysis (MBA) dataset. This dataset contains data from a retailer, where each user’s purchasing transactions and the bought items are recorded. From this dataset, I will provide you the purchase information from about 100 (2.5 percent) users, per-user’s purchasing data has been grouped by their ID and transferred to a natural language description for your better understanding of their purchasing behaviors.

Your task is to generate representative and accurate 20 user personas according to these users’ purchasing patterns. Please notice that we give you 2 important targets you should consider and optimize:
\begin{itemize}
    \item High Coverage: 
    We hope that your generated persona set can cover as many users as possible. We define the ‘coverage’ as the total number of the users which can be labeled with at least one of your generated persona set.

    \item High Accuracy: 
    We hope that each of your generated persona has a precise definition. An ambiguous or sweeping persona definition should be avoided.
\end{itemize}

Repeat your task one more time, the goal is to generate a proper set of 20 representative and accurate user personas existing in the data subset and explain them quantitatively. 

For each persona, you should write a corresponding definition, an output example:

\begin{verbatim}
1. Home Comforts Enthusiast
    - Buys items focused on creating a cozy and inviting home atmosphere, 
      such as wicker hearts, chalkboards, vintage decorative pieces, and heart-shaped ornaments.
2. Craft and DIY Hobbyist
    - Often purchases crafting materials, DIY kits, sewing items, plush toys, 
      and bespoke stationery sets for personal projects or to entertain children.
... (20 personas in total)
\end{verbatim}

Now considering the user purchasing data given below: ...

\end{tcolorbox}
\vspace{-10pt}
\caption{Case study on initial persona set generation (Take MBA as an example) - Step 1.}
\label{fig:persona-gen-step1}
\end{figure*}

\section{Details of Persona Generation and Labelling}\label{app:b}
\stitle{Persona generation}
To establish a comprehensive and representative set of customer personas, we leverage LLMs with prompt engineering techniques. Using the \mba dataset as an example, we provide the templates in Figure~\ref{fig:persona-gen-step1},~\ref{fig:persona-gen-step2} and~\ref{fig:persona-gen-step3}. In the first step, we extract the purchase histories from 100 randomly selected customers and instruct the LLM to generate 20 persona candidates, each with a description. We then repeat this process 40 times to create 40 sets of persona candidates. In the second step, we sample 5 of these 40 sets and prompt the LLM to provide an improved set of 20 personas. This process is repeated 8 times, producing 160 polished persona candidates. Finally, we feed these candidates into the third template to obtain the final set of 20 personas.

\stitle{Persona labeling}
To illustrate the LLM-based persona labeling, we present the detailed prompts in Figures~\ref{fig:persona-label-instruction},~\ref{fig:persona-label-input}, and~\ref{fig:persona-label-result}. First, we instruct the LLM to specify the requirements and format for this task. Next, for each prototype user, we serialize their historically purchased items and purchase frequencies into natural language and input them into the template. Notably, this template reminds the LLM to select personas only from the provided list, resulting in the persona representations shown in the last subfigure.
% response to the consistency problem
To evaluate labeling consistency, we perform three independent rounds of labeling on customers in \mba using GPT-4-Turbo with the default temperature and the same persona set. On average, 33\% of a user’s assigned personas appear in all three runs, while 42\% appear only once. This suggests that the labeling process contains noise, and more robust downstream algorithms may help better leverage these persona labels.

\begin{figure*}[b!]
    \centering
    \begin{tcolorbox}[enhanced,arc=3mm,frame style image=blueshade.png,colback=gray!5!white,
    overlay={\node[anchor=south east,font=\bfseries\itshape\huge,xshift=-5mm,yshift=1mm] at (frame.south east){Step 2};}]
    \textbf{System Prompt:} You are an assistant skilled at reading, observing and summarizing, capable of finding similar or repeated descriptions of user personas, and good at finding the most representative ones.

    \textbf{User Prompt:} Take a deep breath and work according to the instructions step by step. Now we have 40 persona\_sets, each containing 20 personas, and I will randomly select five persona\_sets, each containing 20 personas, for a total of 100 personas. Your task is to select the 20 most representative personas from these 100 personas and output the results. If you find that the content of a certain group is not 20 personas but less than 20, or even irrelevant information, you should ignore this group of information and only refer to personas in other groups. Note that you may find that the personas you read have some similarities or even some duplications. You need to find these similar or duplicate personas and select the 20 most representative personas accordingly. You should not refer to any information related to the number of occurrences in these personas, as this information is very likely to be unreasonable. You should ignore the occurrence times information and select the 20 most representative personas based only on their descriptions.

    Here are the five sets of results that make up the 100 personas you need to choose from: ...
    \end{tcolorbox}
\vspace{-10pt}
\caption{Case study on initial persona set generation (Take MBA as an example) - Step 2.}
\label{fig:persona-gen-step2}
\end{figure*}

\begin{figure*}[b!]
    \centering
    \begin{tcolorbox}[enhanced,arc=3mm,frame style image=blueshade.png,colback=blue!5!white,
    overlay={\node[anchor=south east,font=\bfseries\itshape\huge,xshift=-5mm,yshift=3mm] at (frame.south east){Step 3};}]
    \textbf{System Prompt:} You are an assistant skilled at reading, observing and summarizing, capable of finding similar or repeated description of user personas, and good at finding the most representative ones.

    \textbf{User Prompt:} Take a deep breath and work according to the instructions step by step. Now we have 8 persona sets, each containing 20 personas, for a total of 160 personas. Your task is to select the 20 most representative personas from these 160 personas and output the results. Note that you may find that the personas you read have some similarities or even some duplications. You need to find these similar or duplicate personas and select the 20 most representative personas accordingly, that is to say, these 20 personas occurs most times and can cover most of them.

    Here are the eight sets of personas that make up the 160 personas you need to choose from: ...
    \end{tcolorbox}
\vspace{-10pt}
\caption{Case study on initial persona set generation (Take MBA as an example) - Step 3.}
\label{fig:persona-gen-step3}
\end{figure*}

\renewcommand{\thefigure}{3.\alph{figure}}
\setcounter{figure}{0}

\vspace{-10pt}
\begin{figure*}[b!]
    \centering
    \begin{tcolorbox}[enhanced,title=User 12358 in MBA, frame style image=blueshade.png,
        colback=yellow!10!white,fonttitle=\bfseries,
        overlay={\node[anchor=south east,font=\bfseries\itshape\huge,xshift=-5mm,yshift=5mm] at (frame.south east){Instruction};}]
    \textbf{System Prompt}: Now you are an intelligent e-commerce domain assistant. You are skilled at summarizing, and capable of assigning high-level consumer personas based on a user's purchase behavior.
    
    \bigskip
    
    \textbf{User Prompt:} Take a deep breath and work according to the instructions step by step. 
    
    Your goal is to identify users' shopping behaviors based on products they have bought and label them with a given set of personas. You need to select at least one persona, at most 5 personas from our given persona list. But make sure that for each assignment you should find strong evidence in their purchase transactions. Please keep the procedure as accurate as possible.
    
    Please provide the output in json format. Prefer to return arrays instead of comma separated strings. The following is an explanation of your return format:
    \begin{verbatim}
    {"user_number": ["Persona1", "Persona2", "Persona3"]}
    {"user_number": ["Persona1", "Persona2", "Persona3"]}
    \end{verbatim}
    
    And here is a specific example:
    \begin{verbatim}
    {"12346": [ "Vegan/Vegetarian", "High-Protein Shopper", "Pet Owner"]}
    \end{verbatim}
    
    In the case that you feel there does not exist any suitable persona from the given list that can properly describe a user's purchasing behavior, you can label the user as an 'unrepresentable' user as the following example:
    \begin{verbatim}
    {"12999": ["Unrepresentable"]}
    \end{verbatim}
    \end{tcolorbox}
    \vspace{-10pt}
    \caption{Case study on user persona generation (take user 12358 in MBA as an example) - Instruction.}
    \label{fig:persona-label-instruction}
\end{figure*}

\begin{figure*}[b!]
    \centering
    \begin{tcolorbox}[enhanced,arc=3mm,frame style image=blueshade.png,colback=gray!5!white,
    overlay={\node[anchor=south east,font=\bfseries\itshape\huge,xshift=-5mm,yshift=2mm] at (frame.south east){Input Prompt};}]
    Here is the persona list you should choose from: \textbf{\textit{[PERSONA LIST]}}

    Remember that the user number (i.e., “user\_number” in the example) should be exactly from the given transaction data, do not make it wrong since it is crucial.

    Here is the data of user 12358's transaction data for you to analyze: 
    
    The user 12358 has totally purchased 13 unique products, we show each product name followed by its purchased times: he bought: FAIRY CAKE DESIGN UMBRELLA, 4 times; CERAMIC STRAWBERRY DESIGN MUG, 24 times; CERAMIC CAKE STAND + HANGING CAKES, 2 times; CERAMIC CAKE DESIGN SPOTTED PLATE, 12 times; DOORMAT FAIRY CAKE, 2 times; EDWARDIAN PARASOL PINK, 12 times; EDWARDIAN PARASOL NATURAL, 24 times; EDWARDIAN PARASOL RED, 24 times; EDWARDIAN PARASOL BLACK, 24 times; STRAWBERRY CERAMIC TRINKET BOX, 12 times; CERAMIC BOWL WITH STRAWBERRY DESIGN, 6 times; POSTAGE, 4 times; CERAMIC STRAWBERRY CAKE MONEY BANK, 36 times. Remind one more time that you can only select from the given 20 personas' list and only use the exactly given persona, you cannot use other words to describe. You do not need to explain how you get the result, so please respond no more than the required format.
    \end{tcolorbox}
    \vspace{-10pt}
    \caption{Case study on user persona generation (take user 12358 in MBA as an example) - Input Prompt.}
    \label{fig:persona-label-input}
\end{figure*}

\begin{figure*}[b!]
    \centering
    \begin{tcolorbox}[enhanced,arc=3mm,frame style image=blueshade.png,colback=blue!5!white,
    overlay={\node[anchor=south east,font=\bfseries\itshape\huge,xshift=-5mm,yshift=2mm] at (frame.south east){Generated Result};}]
    \begin{verbatim}
    {"12358": ["Home Decor Aficionado", "Vintage and Retro Enthusiast"]}
    \end{verbatim}
    \end{tcolorbox}
    \vspace{-10pt}
    \caption{Case study on user persona generation (take user 12358 in MBA as an example) - Generated Result.}
    \label{fig:persona-label-result}
\end{figure*}

% \appendix


% \subsection{Prompts}
% prompts

\end{document}